\documentclass[12pt, letter]{extarticle}

\usepackage[margin=1in]{geometry}
\usepackage{booktabs}
\usepackage{authblk}
\usepackage{amsmath,amsthm,amssymb,bm,amsfonts,dsfont}
\usepackage{graphicx, psfrag, epsf, subcaption, tikz}
\usepackage{lipsum}
\usepackage{enumerate}
\usepackage{natbib}
\usepackage{setspace}
\usepackage{textcomp}
\usepackage{adjustbox}
\usepackage{etoolbox}
\usepackage[title]{appendix}
\usepackage[hyphens]{url} 

\usepackage{hyperref}
\usepackage{xcolor} 
\definecolor{UWPink1}{HTML}{FFBEEF}
\definecolor{UWPink2}{HTML}{FF63AA}
\definecolor{UWPink3}{HTML}{DF2498}
\definecolor{UWPink4}{HTML}{C60078}
\definecolor{UWGray1}{HTML}{DFDFDF}
\definecolor{UWGray2}{HTML}{A2A2A2}
\definecolor{UWGray3}{HTML}{787878}
\definecolor{UWGray4}{HTML}{000000}
\definecolor{UWGold1}{HTML}{FFFFAA}
\definecolor{UWGold2}{HTML}{FFEA3D}
\definecolor{UWGold3}{HTML}{FFD54F}
\definecolor{UWGold4}{HTML}{E4B429}
\usepackage{lineno}
\usepackage{hyperref}
\usepackage{float} 
\usepackage[utf8]{inputenc}
\hypersetup{
	colorlinks = true,
	linkcolor = purple,
	citecolor = blue,
	filecolor = magenta, 
	urlcolor  = cyan
}
\usepackage{colonequals}

\newcommand{\bX}{\mathbf{X}}

\newcommand{\inp}[2]{\left\langle #1, #2 \right\rangle}
\newcommand{\norm}[1]{\left\|#1\right\| }
\newcommand{\abs}[1]{\left| #1 \right|}
\newcommand{\N}{\mathcal N}
\newcommand{\EXp}{\mathbb{E}}
\newcommand{\E}{\mathbb{E}}
\newcommand{\hmu}{\hat{\mu}}
\newcommand{\xih}{X_{i+h}}
\newcommand{\xihmu}{X_i-\hat{\mu} }
\newcommand{\xihhmu}{X_{i+h}-\hat{\mu} }

\newcommand{\nxi}{\left\lVert X_i \right\rVert}
\newcommand{\nxih}{\left\lVert X_{i+h}\right\rVert}
\newcommand{\nhmu}{\left\lVert \hat{\mu} \right\rVert}
\newcommand{\nxihmu}{\left\lVert X_i - \hat{\mu}  \right\rVert}
\newcommand{\nxihhmu}{\left\lVert X_{i+h} - \hat{\mu}  \right\rVert}

\newcommand{\san}{\sum_{\{i:(X_i,X_{i+h})\in A_n\}}}
\newcommand{\sanc}{\sum_{\{i:(X_i,X_{i+h})\in A_n^c\}}}

\newcommand{\HH}{\mathbb{H}}

\newcommand{\1}[1]{\mathds{1}_{#1}}

\newcommand{\id}{\overset{D}{\to}}

\newcommand{\ias}{\overset{\text{a.s.}}{\to}}

\newcommand{\hip}{\mathcal{H}^\prime}

\newtheorem{definition}{Definition}
\newtheorem{theorem}{Theorem}

\newtheorem{lemma}{Lemma}
\newtheorem{assumption}{Assumption}

\allowdisplaybreaks

\begin{document}
	
\def\spacingset#1{\renewcommand{\baselinestretch}%
	{#1}\small\normalsize} \spacingset{1}

{
	\title{\bf Functional Spherical Autocorrelation: A Robust Estimate of the Autocorrelation of a Functional Time Series}
	\author[1]{Chi-Kuang Yeh\thanks{chi-kuang.yeh@uwaterloo.ca}}
	\author[1]{Gregory Rice\thanks{grice@uwaterloo.ca}}
	\author[1,2]{Joel A. Dubin\thanks{jdubin@uwaterloo.ca}}
	\affil[1]{Department of Statistics and Actuarial Science, University of Waterloo, Canada}
	\affil[2]{School of Public Health Sciences, University of Waterloo, Canada}
	
	\setcounter{Maxaffil}{0}
	\renewcommand\Affilfont{\itshape\small}
	\maketitle
}

\begin{abstract}
	We propose a new autocorrelation measure for functional time series that we term ``spherical autocorrelation". It is based on measuring the average angle between lagged pairs of series after having been projected onto the unit sphere. This new measure enjoys several complimentary advantages compared to existing autocorrelation measures for functional data, since it both 1) describes a notion of ``sign" or ``direction" of serial dependence in the series, and 2) is more robust to outliers. The asymptotic properties of estimators of the spherical autocorrelation are established, and are used to construct confidence intervals and portmanteau white noise tests. These confidence intervals and tests are shown to be effective in simulation experiments, and demonstrated in applications to model selection for daily electricity price curves, and measuring the volatility in densely observed asset price data. 
\end{abstract}

\noindent%

\renewcommand{\baselinestretch}{0.5}\footnotesize
\spacingset{1.2} 
	
\section{Introduction}

Functional data analysis as a field has grown considerably over the past three decades, likely owing to increasing interest in analyzing high-dimensional data that arise from continuously observing processes over various domains, including time, space, and frequency. We refer the reader to \cite{HKbook} and \cite{ramsay:silverman:2005:FDA} for textbook length treatments of functional data analysis. In many cases of interest, functional data are collected sequentially over time. For example, continuous observations of the electricity price in a given region might be interpolated and fashioned into daily electricity price curves.  Such series of functional data objects are referred to as  \textit{functional time series}, and methods for analyzing and modelling functional time series have also been steadily developed in recent years; see \cite{bosq:2000},  \cite{hormann:kokoszka:2010}, and Chapter 8 of \cite{kokoszka:2017:FDA-book} for an overview of functional time series analysis. 

When analyzing, or evaluating a model for, any time series, one often begins by computing and plotting a measure of the sample autocorrelation of the series, along with corresponding confidence intervals computed assuming the series forms a strong white noise. This is usually done in order to evaluate the serial dependence structure of the series or model residuals in order to inform further modelling, or serve as validation that a given time series model appears to fit the data well.

In terms of estimating the serial dependence structure of a functional time series, available methods to date are based on measuring the magnitude of the lagged autocovariance operator of the series. In order to fix ideas, suppose that $\{X_i(t),\;\; i \in \mathbb{Z},\;\; t\in [0,1]\}$ is a sequence of stochastic processes whose sample paths are in $L^2[0,1]$, the space of real-valued square-integrable functions defined on the unit interval. For example, $X_i(t)$ might denote the electricity price in a certain region on day $i$ at time $t$, where $t$ in normalized to lie in $[0,1]$. Conceptually a functional time series can be thought of as an observed stretch of such a sequence of length $n$, $X_1,...,X_n$.  As defined in, for example, \cite{panaretos:tavakoli:2012}, when the series $\{X_i(t),\;\; i \in \mathbb{Z},\;\; t\in [0,1]\}$ is stationary, the {\it autocovariance kernel at lag $h$} of the series is

\[
	C_h(t, s)=\EXp{\left[X_{0}(t)-\mathbb{E} X_{0}(t)\right]\left[X_{h}(s)- \mathbb{E}X_{h}(s)\right]}.
\]
Using the standard $L^2$-norm to measure the magnitude of $C_h$, a {\it functional autocorrelation function} (fACF) may be defined as the mapping $h \mapsto \rho_h^{(F)}$, where 
\[
	\rho_{h}^{(\text{F})}=\frac{\left\|C_{h}\right\|_2}{\int C_{h}(t, t) d t}, \mbox{ and } 	\left\|C_{h}\right\|^{2}_2=\iint C_{h}^{2}(t, s) dtds.
\]
\cite{mestre:2021:facf} considered this definition of the fACF, and studied its use in performing model selection for functional SARIMAX models. \cite{kokoszka:2017:autocovariance-ftsa-conditional} also considered the properties of this fACF assuming the underlying functional time series is conditionally heteroscedastic. 

Although this fACF is useful in many situations, it has several drawbacks. Among them are that $\rho_h^{(F)}$  describes only the magnitude of $C_h$, and is always non-negative. Any information about the ``sign" or ``direction" of the correlation in the series is lost in computing $\rho_h^{(F)}$. Additionally, estimating  $\rho_h^{(F)}$ is highly prone to the influence of outliers. Although the definition of the fACF only presumes the existence of two moments of the underlying series, i.e. $\E \|X_i\|^2 < \infty$, in order to perform inference on $\rho_h^{(F)}$, one often requires four moments, $\E \|X_i\|^4 < \infty$. Worse still, in many cases of interest, for instance in the context of fitting volatility models to functional time series, one often wishes to estimate the autocorrelation structure of transformations of the original series. A popular transformation to consider in this setting is $X_i^2$, and in this case the implicit number of moments required in order for the usual confidence intervals for $\rho_h^{(F)}$ to be valid is   $\E \|X_i\|^8 < \infty$.

To address some of these shortcomings, we propose a new autocorrelation measure for functional time series that we dub the {\it spherical autocorrelation function}. It is constructed by computing the average angle, as encoded by the inner product of projections of the series onto the unit sphere in $L^2[0,1]$,  between lagged pairs in the series that have been centred by the spatial (geometric) median of the series. This measure both captures ``signed" serial dependence in the series, and is robust in the sense that it does not require high moment conditions to define and estimate. We show that this new autocorrelation function can be consistently estimated when the underlying series is stationary, and satisfies the central limit theorem assuming the series is a strong white noise.

The main theoretical challenge that we overcome in doing so lies in handling the effect of estimating the spatial median, since the natural estimator of the spherical autocorrelation is not easily linearized in terms of the spatial median, and existing spatial median estimators are defined in terms of iterative optimization procedures, and do not have closed-form expressions; see e.g. \cite{gervini:2008:robust-fpca} and \cite{cardot:2013}. We establish though that under primitive conditions, which are satisfied by existing functional spatial median estimators, the estimator that we propose for the spherical autocorrelation is still consistent and asymptotically Gaussian when the underlying series is a strong white noise. 

These results are used to construct asymptotically valid confidence intervals and portmanteau white noise tests for the spherical autocorrelation function. As confirmed in simulation experiments and two real data applications, these confidence intervals and tests appear to perform well in many settings, and are useful in exploring the dependence structure, and performing model selection, with functional time series.  

There are a number of papers related to this work that generally address the problems of estimating the correlation between functional data objects, and robust covariance estimation with functional data.  \cite{leurgans:1993:canonical} studied the canonical correlation between paired functional data, see also page 28 in \cite{ramsay:silverman:2005:FDA}, and \cite{boente:Kudraszow:2022:canonical} put forward a robust estimator for the canonical correlation of functional data. \cite{boente:2019:spatial} develop a robust, spherical covariance operator for functional data, which is also based on centring the data using the spatial median.

The rest of the paper is organized as follows. Section \ref{sec-method} gives a definition of the spherical autocorrelation and its natural estimator. That section also includes statements of the asymptotic properties of the estimator, as well as definitions of confidence intervals and portmanteau tests.  The practical implementation of the methods are discussed in Section \ref{sec-simulation}, as well as the results of several simulation experiments. Two  data applications are provided in Section \ref{sec-applicatipn}, and we close the paper with concluding remarks in Section \ref{sec-conclusion}. All proofs and some additional implementation details and simulation results are placed in an Appendix following these main sections.

\section{Problem statement and methodology}\label{sec-method}

We begin by introducing some notation that will be used throughout the paper. We consider a functional time series $\{X_i\}_{i \in \mathbb{Z}} := \{X_i(t): i\in\mathbb{Z}, \; t\in\mathcal{T}\}$, where $\mathcal{T}$ is a compact interval, from which we assume that we have observed a stretch of length $n$, $X_1,...,X_n$. Without loss of generality, we assume the argument $t$ is re-scaled so that  $\mathcal{T}=[0,1]$. Here each variable $X_i$ is viewed as a stochastic process whose sample paths lie in the Hilbert space $\mathbb{H}=L^2[0,1]$ of real valued square integrable functions equipped with the inner product  
\[
\inp{f}{g} = \int f(t)g(t)dt,
\]
and corresponding norm $\|\cdot \| = \sqrt{\langle \cdot , \cdot \rangle}$, where $\int=\int_0^1$. Although the methods and results below would hold equally well when $\mathbb{H}$ is any separable Hilbert space, we take $\mathbb{H}=L^2[0,1]$  since this both covers the applications that we consider, and also coincides with the space considered for many existing methods for robust estimation of the centre and covariance of functional data; see \cite{gervini:2008:robust-fpca} and \cite{boente:2019:spatial}. 

\subsection{Definition of spherical autocorrelation}
In the event that the series $\{X_i\}_{i\in Z}$ is strictly stationary, our main goal is to quantify the autocorrelation present in the series using the observations $X_1,...,X_n$. The method that we propose is based on examining the angle between suitably centred, lagged pairs of the series $X_i$ and $X_{i+h}$. This angle is encoded by the inner product of the projections of these pairs onto the unit sphere. We denote the projection of $x\in \mathbb{H}$ onto the unit sphere as $S(x) := x/\norm{x}$ for $x\ne 0$, and $S(0) = 0$. $S(\cdot)$ is often referred to as the \textit{spatial sign}. The centred and projected observations may  then be written as 
\begin{equation}
	S(X_i-\mu) = \frac{X_i  - \mu }{\norm{X_i-\mu}}, \quad i = 1,\cdots,n,
\end{equation}
where $\mu$ is a suitably defined ``centre" of the process $X_i$, which we discuss in detail below. Note that $S(x)$ is bounded as it is easy to see that $\|S(x)\|\le 1$. The following defines a new notion of autocorrelation for a functional time series.
\begin{definition}{\rm Suppose the series $\{X_i\}_{i\in Z}$ is strictly stationary. Given a fixed time lag $h$, we define the {\it spherical autocorrelation} at lag $h$ as
	\begin{equation}
		\rho_h^{\text{(S)}} = \rho_h = \EXp{\inp{S(X_0-\mu)}{S(X_h-\mu)}}.
	\end{equation}
	The {\it functional spherical autocorrelation function} (abbreviated fSACF) is the function  $h \mapsto \rho_h^{(S)}$ for $h\in \{0,1,...\}$. }
\end{definition} 

It is immediate from its definition that the fSACF shares a number of similar features with the classic autocorrelation function of a univariate time series: i)  $-1 \le \rho_h \le 1$ , ii) $\rho_0=1$, and iii) the fSACF is symmetric around the origin, i.e., $\rho_h = \rho_{-h}$. Another clear consequence of its definition is that $\rho_h$ is well-defined for any stationary process in $\mathbb{H}$. In particular, no moment conditions are required in defining the fSACF. 

In order to maintain its analogy to other classic notions of autocorrelation and aid in its interpretation, it would also be an appealing feature of the fSACF if when the variables are independent at lag $h$, i.e. when $X_0$ and $X_{h}$ are independent, $\rho_h=0$. It turns out that this property depends primarily on how the observations are centred by $\mu$. A natural candidate is to take $\mu(t)= \E X_i(t)$, the mean function of the observations, however this does not lead to a fSACF with this property.

What appears instead to be the correct ``centre" in order to preserve this property is take $\mu$ to be the functional spatial median: 

\begin{definition}\label{sp-def} {\rm  The {\it spatial median} (also sometimes referred to as the {\it geometric median}) $\mu$ of a random variable $X \in \HH$ is any element of $\mathbb{H}$  satisfying 
	\[
	\mu  = \underset{y \in \HH} {\arg \min }~ \EXp{\norm{X - y} - \norm{X}}.
	\]	}
\end{definition}
See \cite{kemperman:1987}, \cite{gervini:2008:robust-fpca},  \cite{cardot:2013},  and  \cite{boente:2019:spatial}. As noted in equation (2) of \cite{gervini:2008:robust-fpca}, when $P(X=\mu)=0$, the definition of the spatial median implies that $\E S(X-\mu)=0$. As a result when $X_0$ and $X_h$ are independent and $\mu$ is the spatial median, we have by Fubini's theorem that $\rho_h = \langle \E S(X_0-\mu), \E S(X_h-\mu) \rangle=0,  $ so long as  $P(X_i=\mu)=0$. 

Going forward we thus take $\mu$ in the definition of $\rho_h$ to be the functional spatial median as in Definition \ref{sp-def}.

\subsection{Inference for  \texorpdfstring{$\rho_h$}{rho} assuming \texorpdfstring{$\mu$}{u} is known}\label{subsec-mean-k}

Although in practice the spatial median $\mu$ must be estimated from the sample in order to produce feasible estimators of $\rho_h$, we think it is useful to begin by considering estimators of $\rho_h$, and their properties, assuming $\mu$ is known. In the next subsection we will consider the properties of these estimators when  $\mu$ is replaced with typical estimators of the spatial median. Presenting the results in this order allows us to separate the properties of the underlying process $\{X_i\}_{i\in \mathbb{Z}}$ that are needed to perform inference on $\rho_h$ from those that are required to handle the effect of estimating $\mu$.    

Assuming $\mu$ is known, a natural estimator of $\rho_h$ is 
\begin{equation} \label{eq-est-dynauto-false}
	\tilde{\rho}_h = \frac{1}{n }\sum_{i=1}^{n-h} \inp{S(X_i-\mu)}{S(X_{i+h}- \mu)}.
\end{equation}

When analyzing a functional time series, we are often interested in evaluating whether the series, perhaps comprised of model residuals, appears to be a white noise, or instead appears to exhibit significant autocorrelation at some lags, which may inform subsequent modelling.   As such our main aim is to approximate the distribution of $\tilde{\rho}_h$ if the underlying series forms a strong white noise, i.e. a sequence of independent and identically distributed functional data objects, and also show that $\tilde{\rho}_h$ is consistent for $\rho_h$ if the underlying series is strictly stationary. 

If the series $\{X_i\}_{i\in \mathbb{Z}}$ is strictly stationary and ergodic, then $\tilde{\rho}_h \stackrel{a.s.}{\to} \rho_h$ by the mean ergodic theorem. Our first result then concerns the asymptotic distribution of $\tilde{\rho}_h$ when the underlying series is a strong white noise. Let 

$$
C_P(t,s) = \E [S(X_0 - \mu )(t)][S(X_0-\mu )(s)], \quad t,s\in[0,1],
$$
denote the covariance kernel of the observations projected onto the unit sphere in $\mathbb{H}$.  

\begin{theorem}\label{thm-1} Suppose that $\{X_i\}_{i \in \mathbb{Z}}$ is a strong white noise in $\HH$ such that  $P(X_0 = \mu)=0$. Then 

	\begin{equation}
		\sqrt{n} \tilde{\rho}_h  \id \N(0,\|C_p\|_2^2),
	\end{equation}
where $\|\cdot\|_2$ is the Hilbert-Schmidt norm of a kernel. Moreover, if for a positive integer $H$ $\tilde{R}_H=(\tilde{\rho}_1,...,\tilde{\rho}_H)^\top$, then	

	\begin{equation}\label{eq-distribution}
		\sqrt{n} \tilde{R}_H  \id \N_H(0,\|C_p\|_2^2 \mathbb{I}_H),
	\end{equation}
	where $\N_H(m,\Sigma)$ denotes an $H$-dimensional normal random vector with mean vector $m$ and covariance $\Sigma$, and $\mathbb{I}_H$ is the $H$-dimensional identity matrix. 
\end{theorem}
The proof of Theorem \ref{thm-1} is given in the Appendix \ref{app-thm1}. This result may be used to justify the asymptotic consistency of tests of the hypotheses:

$\mathcal{H}_{0,h}$: For a fixed positive integer $h$, $\rho_h=0$, and 

$\mathcal{H}^\prime_{0,H}$: For a fixed positive integer $H$, $\rho_h=0$ for all $h\in \{1,...,H\}$. 

In particular, an asymptotically sized $\alpha$ test of $\mathcal{H}_{0,h}$ is to reject when  $|\sqrt{n}\tilde{\rho}_h| > \|C_P\|_2 z_{ \alpha/2}$, where $z_q$ is the $q$ critical value of the standard normal distribution. Similarly, a $1-\alpha$ confidence set for the fSACF assuming the series is a strong white noise is 
\begin{align}\label{conf-def-tilde}
-\frac{z_{1-\alpha/2}}{\sqrt{n}\|C_P\|_2} \le \tilde{\rho}_h  \le \frac{z_{1-\alpha/2}}{\sqrt{n}\|C_P\|_2}. 
\end{align}
Comparing the estimated fSACF $\tilde{\rho}_h$ to its 95\% confidence bounds assuming the underlying series follows a strong white noise leads to a simple visual summary that may be used to measure at a glance the serial dependence structure of a functional time series; see e.g. Figure \ref{fig-white-noise} below.  

Testing $\mathcal{H}^\prime_{0,H}$ is often referred to as ``portmanteau testing". As a result of Theorem \ref{thm-1}, the test statistic is
$$
\tilde{Q}_{n,H} = \|\sqrt{n}\tilde{R}_H\|^2_E = n \sum_{h=1}^H \tilde{\rho}_h^2 \id \|C_P\|^2 \chi^2(H), 
$$
where $\|\cdot\|_E$ is the standard Euclidean norm of a vector, and $\chi^2(H)$ denotes a $\chi^2$ random variable with $H$ degrees of freedom. An approximate P-value of a test of $\mathcal{H}_{0,H}^\prime$  may be computed as $p= P( \|C_P\|^2 \chi^2(H) > \tilde{Q}_{n,H})$.

\subsection{Inference for \texorpdfstring{$\rho_h$}{rho} assuming \texorpdfstring{$\mu$}{u} is unknown}

In practice of course the spatial median $\mu$ is typically not known, and must instead be estimated from the sample. Given its definition, it is natural to estimate $\mu$ with
\begin{align}\label{med-est-def}
\hat{\mu}_{ideal}   = \underset{\mu \in \HH} {\arg \min } \sum_{i=1}^n\norm{X_i - \mu}.
\end{align}
Although there is no closed-form solution for the estimator satisfying \eqref{med-est-def}, several authors have proposed iterative procedures to approximate $\hat{\mu}_{ideal}$, and have shown that these lead to consistent estimators with quantifiable convergence rates. For example, \cite{gervini:2008:robust-fpca} proposes an approach to approximately solve \eqref{med-est-def} using Gateaux differentials, and \cite{cardot:2013} propose an approach based on stochastic gradient decent. Both of these approaches result in an estimator $\hat{\mu}$ approximately satisfying \eqref{med-est-def}. Using $\hat{\mu}$ in place of $\mu$ in \eqref{eq-est-dynauto-false} leads to the feasible estimator 

\begin{equation} \label{eq-est-dynauto-true}
	\hat{\rho}_h = \frac{1}{n }\sum_{i=1}^{n-h} \inp{S(X_i-\hat{\mu})}{S(X_{i+h}- \hat{\mu})}.
\end{equation}
We now aim to establish conditions under which $\hat{\rho}_h$ also satisfies Theorem \ref{thm-1} and is asymptotically consistent to $\rho_h$. Or, in other words, conditions under which the effect of estimating the spatial median $\mu$ is asymptotically negligible. We begin by establishing an analogue of Theorem \ref{thm-1} under the following additional assumptions.    

\begin{assumption}\label{as-tight} The sequence $\sqrt{n}(\hat{\mu}- \mu )$ is uniformly tight in $\HH$.
\end{assumption}

\begin{assumption}\label{as-moment-2} 
There exists constants $C_1,C_2>0$ so that for all $u \in \HH$ with $\|u \|\le C_1$, $\E \|X_i - (\mu + u )\|^{-2} \le C_2$.
\end{assumption}
For a definition of uniform tightness in a metric space see \cite{billingsley:1968}.  Assumption \ref{as-moment-2} coincides with Assumption A3 in \cite{cardot:2013}, where it is discussed extensively. It holds, for example, for Gaussian processes under mild conditions, as well as for many other processes possessing exponentially decaying small-ball probabilities; see  \cite{lishao:2001,nazarov:2009}. The stochastic gradient decent-based estimator of \cite{cardot:2013} also satisfies Assumption \ref{as-tight} under Assumption \ref{as-moment-2} and an additional assumption entailing that the spatial median is unique. The estimator of the spatial median proposed in \cite{gervini:2008:robust-fpca}  also satisfies Assumption \ref{as-tight} under similar conditions. 

\begin{theorem}\label{thm-est-white-con} 
Suppose that $\{X_i\}_{i \in \mathbb{Z}}$ is a strong white noise in $\HH$, and that Assumptions \ref{as-tight} and \ref{as-moment-2} hold. Then $\hat{\rho}_h$ satisfies Theorem \ref{thm-1} in place of $\tilde{\rho}_h$. 
\end{theorem}

Theorem \ref{thm-est-white-con} shows that under Assumptions \ref{as-tight} and \ref{as-moment-2}, the same confidence sets and portmanteau tests based on $\tilde{\rho}_h$ remain asymptotically valid when based on $\hat{\rho}_h$ instead.  

We now turn to the asymptotic consistency of $\hat{\rho}_h$, which we establish under the following assumptions.  

\begin{assumption}\label{ass1} $\hmu \ias \mu$ as $n\to \infty$. 
\end{assumption}

\begin{assumption}\label{ass2}
	For any $h\ge 1$, $\EXp{\norm{X_0}^{-1}\norm{X_h}^{-1}}< \infty$. 
\end{assumption}

\begin{assumption}\label{ass3}
	  $\EXp{\norm{X_0}^{-1}} <\infty$ . 
\end{assumption}

\begin{theorem}\label{thm-cons-2} Suppose $\{X_i\}_{i \in \mathbb{Z}}$ is strictly stationary and ergodic, and that Assumptions \ref{ass1}-\ref{ass3} hold. Then $\hat{\rho}_h \stackrel{a.s.}{\to} \rho_h$ as $n \to \infty$. 
    
\end{theorem}

Assumption \ref{ass1}-\ref{ass3} are analogous in a time series setting of the assumptions in  Theorem 1 of \cite{boente:2019:spatial}, and the proof makes use of some of the same ideas. Theorem \ref{thm-cons-2} implies that when $\rho_h \ne 0$ for some $h\in \{1,...,H\}$, both tests of $\mathcal{H}_{0,h}$ and $\mathcal{H}^\prime_{0,H}$ described at the end of Section \ref{subsec-mean-k}  are consistent when based on $\hat{\rho}_h$.  

\section{Implementation and simulation studies}\label{sec-simulation}

In this section we briefly provide some details on implementation, and then present the results of several simulation studies that aimed to evaluate the proposed inferential procedures for the fSACF.

In the simulations and applications below, we assume that all functional data are observed on a common grid $t_j \in [0,1]$, $j=1,...,M$. In our applications to electricity price data and intraday asset price data below, each functional data object is observed at $M=24$ and $M=78$ points, respectively, and for all simulated examples we generated the data on $M=101$ equally spaced points. In these cases all inner products and norms involved in producing $\hat{\rho}_h$ can be estimated by simple Riemann integration. For instance, with $t_0=0$ 
$$
\|X_i-\hat{\mu}\|^2 \approx \sum_{j=1}^M [X_i(t_j)-\hat{\mu}(t_j)]^2[t_j-t_{j-1}]$$
and 
\begin{align*}
\langle S(X_i - \hat{\mu}),S(X_{i+h} - \hat{\mu}) \rangle &\approx \frac{1}{ \|X_i-\hat{\mu}\|\|X_{i+h}-\hat{\mu}\|} \times \\
&\sum_{j=1}^M [X_i(t_j)-\hat{\mu}(t_j)][X_{i+h}(t_j)-\hat{\mu}(t_j)](t_j-t_{j-1}).   
\end{align*}
If the functional data are not observed on a common grid, then one may use interpolation techniques to estimate the integrals required to compute $\hat{\rho}_h$; see e.g. Chapter 2 of \cite{ramsay:silverman:2005:FDA}.    

In order to estimate $\|C_P\|_2$,  we use the estimator 
\[
\hat{C}_P(t,s) = \frac{1}{n}\sum_{j=1}^n S(X_j(t)-\hat{\mu}(t)) S(X_j(s)-\hat{\mu}(s)),
\]
and estimate its norm based on the discrete sample with 
\[
 \|\hat{C}_P^2\|_2^2 \approx \sum_{j,k=1}^M  \hat{C}_P^2(t_j,t_k)[t_j-t_{j-1}][t_k-t_{k-1}].
\]

Additional simulations evaluating the accuracy of these estimators are provided in Appendix \ref{sim-app}, which suggested that they work well so long as the resolution of the data is not too low. With these estimators, the confidence sets for $\hat{\rho}_h$ assuming the series is a strong white noise, and p-values for the portmanteau test, may be calculated respectively as 

\begin{align}\label{conf-def}
-\frac{z_{1-\alpha/2}}{\sqrt{n}\|\hat{C}_P\|_2} \le \hat{\rho}_h  \le \frac{z_{1-\alpha/2}}{\sqrt{n}\|\hat{C}_P\|_2}, 
\end{align}
and 
\begin{align}\label{port-test}
\hat{Q}_{n,H} = n \sum_{h=1}^H \hat{\rho}_h^2, \;\; p=P( \|\hat{C}_P\|^2 \chi^2(H) > \hat{Q}_{n,H}). 
\end{align}

In all examples below $\mu$ was estimated using the estimator described in \cite{gervini:2008:robust-fpca}. 

We now turn to presenting the results of several simulation experiments. Below we refer to the functional autocorrelation function of \cite{mestre:2021:facf} as the \textit{fACF}. All the computations are performed using \textit{R}, version 4.1.2 \cite{R412}. 

\subsection{White noise processes}
We simulated several strong white noise (WN) processes,  and examined their fACF and fSACF. The examples that we considered were

\begin{enumerate}
	\item Standard Brownian motion (BM): $X_i(t) = W_i(t)$, where $\{W_i(t), \; t\in[0,1]\}_{i \in \mathbb{Z}}$ are iid standard Brownian motions. 
	\item Brownian bridge (BB): $X_i(t)=B_i(t)$, where   $B_i(t) = W_i(t) - tW_i(1)$.
	\item Fourier-Cauchy process (F): $X_i(t)=F_i(t)$, where $F_i(t) = Z_{1,i}+ \\\sum_{k=1}^3 [Z_{2k,i}\cos(2\pi k t) + Z_{2k+1,i}\sin(2\pi k t) ] $ and $Z_{j,i}$ are a family of independent and identically distributed standard Cauchy random variables. 
	\item \sloppy B-spline Exponential process (BS): $X_i(t) = Bs_i(t)$ where $Bs_i(t) = \sum_{k=1}\varepsilon_{i,k}B_k(t)$ in which $\varepsilon_{i,k}$ is a family of independent and identically distributed  exponential random variables with mean 1, and $B_k(t)$ is the k-th orthogonal cubic B-spline basis function.
\end{enumerate} 

The top panel of Figure \ref{fig-white-noise} displays  50 iid realizations of each WN process along with the estimated spatial median.  The estimated fSACF for each process is displayed in the bottom row of Figure \ref{fig-white-noise}, along with blue lines indicating the confidence intervals in \eqref{conf-def}.

In order to evaluate the coverage properties of the proposed confidence intervals for the fSACF, each WN process was independently generated 1000 times for sample sizes $n\in \{100,250,500,1000,2000\}$.  Table \ref{tbl-SWN} shows the proportion of estimators $\hat{\rho}_h$ and $\tilde{\rho}_h$ that were not contained in the $1-\alpha$ confidence intervals \eqref{conf-def} and \eqref{conf-def-tilde}, respectively, for $\alpha \in \{10\%, 5\%, 1\%\}$, and for several values of $h$. These results suggest that, for the sample sizes and processes considered, the confidence intervals derived from Theorems \ref{thm-1} and \ref{thm-est-white-con} have approximately correct coverage. We notably also observed this for the Fourier-Cauchy process, as well as the exponential process, which lends some evidence to the statement that the fSACF is robust against low-order moments/outliers and skewness in the data in practice. For smaller $n$ and larger values of $h$, we noticed that the intervals tended to have somewhat larger than nominal coverage, which we think may be attributed to the fact that the estimators $\hat{\rho}_h$ and $\tilde{\rho}_h$ are normalized by $1/n$ rather than $1/(n-h)$. In unreported simulations when we normalized by $1/(n-h)$ instead we noticed the opposite effect.  Additionally, we observed that the effect of estimating the spatial median appeared to be negligible in all cases.  

\begin{figure} [tbp]
	\centering
	\begin{subfigure}[b]{0.235\textwidth}
		\centering
		\includegraphics[width=\textwidth]{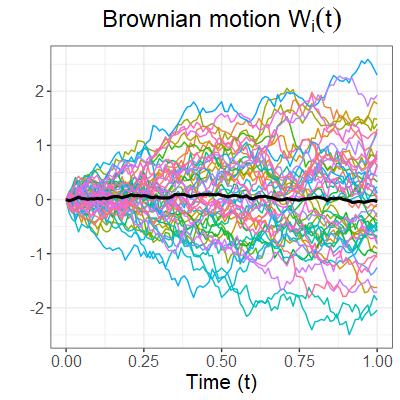}
	\end{subfigure}
	\hfill
	\begin{subfigure}[b]{0.235\textwidth}
		\centering
		\includegraphics[width=\textwidth]{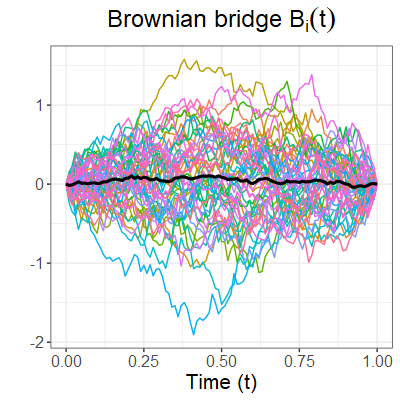}
	\end{subfigure}
	\hfill
	\begin{subfigure}[b]{0.235\textwidth}
		\centering
		\includegraphics[width=\textwidth]{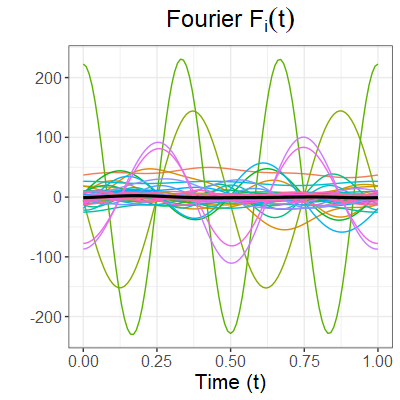}
	\end{subfigure}
	\hfill
	\begin{subfigure}[b]{0.235\textwidth}
		\centering
		\includegraphics[width=\textwidth]{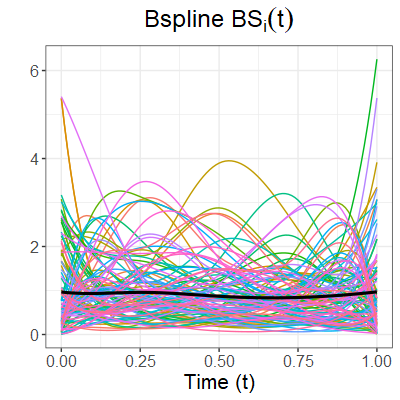}
	\end{subfigure}
	\hfill 
	\begin{subfigure}[b]{0.235\textwidth}
		\centering
		\includegraphics[width=\textwidth]{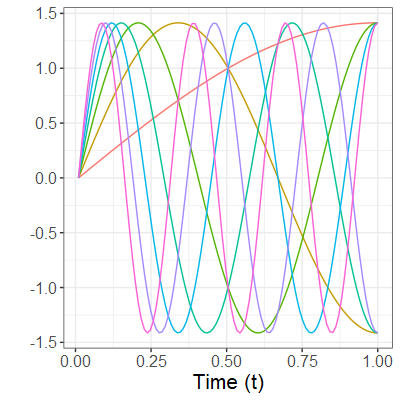}
	\end{subfigure}
	\hfill
	\begin{subfigure}[b]{0.235\textwidth}
		\centering
		\includegraphics[width=\textwidth]{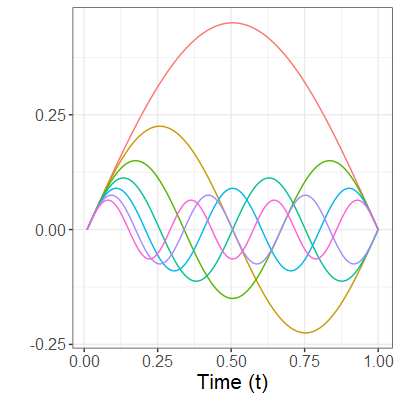}
	\end{subfigure}
	\hfill
	\begin{subfigure}[b]{0.235\textwidth}
		\centering
		\includegraphics[width=\textwidth]{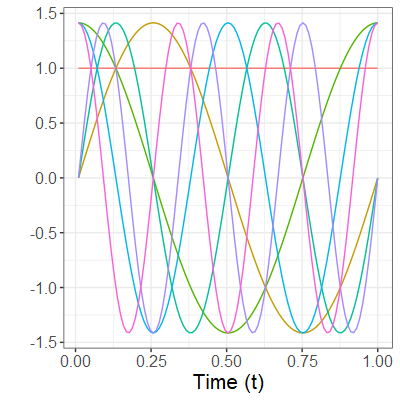}
	\end{subfigure}
	\hfill
	\begin{subfigure}[b]{0.235\textwidth}
		\centering
		\includegraphics[width=\textwidth]{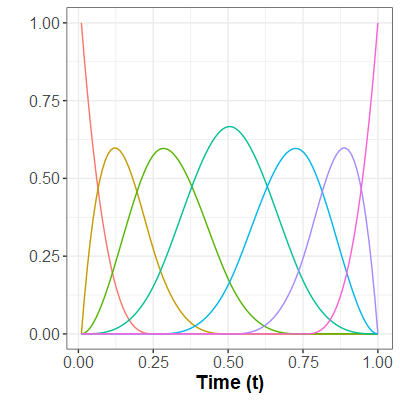}
	\end{subfigure}
	\hfill
	\begin{subfigure}[b]{0.235\textwidth}
		\centering
		\includegraphics[width=\textwidth]{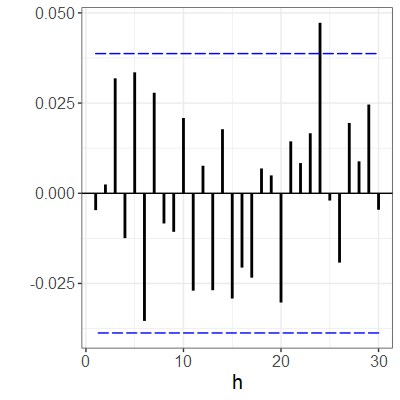}
	\end{subfigure}
	\hfill
	\begin{subfigure}[b]{0.235\textwidth}
		\centering
		\includegraphics[width=\textwidth]{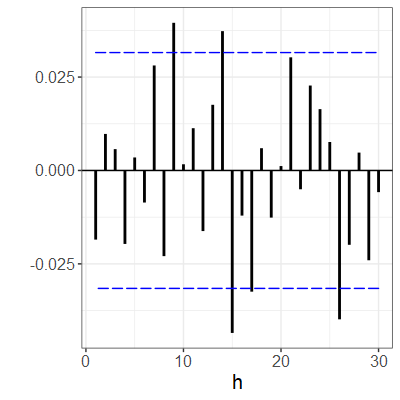}
	\end{subfigure}
	\hfill
	\begin{subfigure}[b]{0.235\textwidth}
		\centering
		\includegraphics[width=\textwidth]{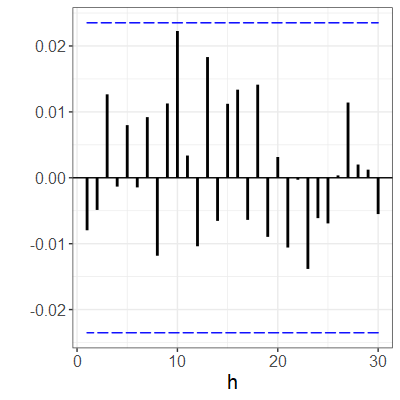}
	\end{subfigure}
	\hfill
	\begin{subfigure}[b]{0.235\textwidth}
		\centering
		\includegraphics[width=\textwidth]{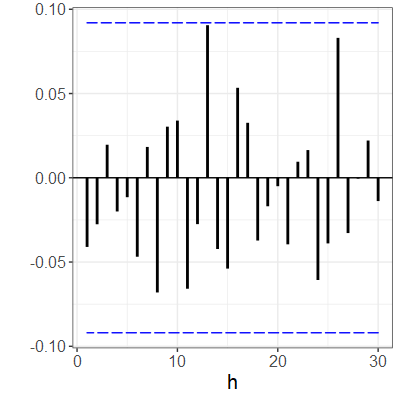}
	\end{subfigure}
	\caption{\textbf{Top}: 50 iid curves of each strong noise process with the spatial median estimate from \cite{gervini:2008:robust-fpca} (in black); \textbf{Middle}: The leading 7 basis functions in the Karhunen-Lo\'eve expansion of each process; \textbf{Bottom}: The estimated fSACF of each series for lags $h\in \{1,...,30\}$ with 95\% confidence bounds as in \eqref{conf-def} (in blue). }\label{fig-white-noise}
\end{figure}

\begin{table}[ht!]
\centering
\caption{The empirical rate among 1000 independent simulations of four WN processes that the $1-\alpha$ confidence intervals \eqref{conf-def-tilde} and \eqref{conf-def} do not contain the estimators $\tilde{\rho}_h$ and $\hat{\rho}_h$ for $\alpha \in \{10\%, 5\%, 1\%\}$, $n\in \{100,250, 500,1000,2000\}$ and several values of $h$.}
\label{tbl-SWN}
\begin{adjustbox}{max width=\textwidth}
\begin{tabular}{@{}llllllllllllllllll@{}}
\toprule
          & n    & h  & \multicolumn{6}{c}{Significance Level $\alpha$}                                                           &                      & n                    & h                    & \multicolumn{6}{c}{Significance level}                                                                    \\ \midrule
          &      &    & \multicolumn{2}{c}{10\%}          & \multicolumn{2}{c}{5\%}           & \multicolumn{2}{c}{1\%}           & \multicolumn{1}{c}{} & \multicolumn{1}{c}{} & \multicolumn{1}{c}{} & \multicolumn{2}{c}{10\%}          & \multicolumn{2}{c}{5\%}           & \multicolumn{2}{c}{1\%}           \\
Statistic &      &    & $\tilde{\rho}_h$ & $\hat{\rho}_h$ & $\tilde{\rho}_h$ & $\hat{\rho}_h$ & $\tilde{\rho}_h$ & $\hat{\rho}_h$ &                      &                      &                      & $\tilde{\rho}_h$ & $\hat{\rho}_h$ & $\tilde{\rho}_h$ & $\hat{\rho}_h$ & $\tilde{\rho}_h$ & $\hat{\rho}_h$ \\ \midrule
BM        & 100  & 1  & 0.104            & 0.108          & 0.052            & 0.055          & 0.011            & 0.013          & BB                   & 100                  & 1                    & 0.086            & 0.093          & 0.048            & 0.05           & 0.008            & 0.013          \\
          &      & 5  & 0.116            & 0.112          & 0.053            & 0.062          & 0.011            & 0.010          &                      &                      & 5                    & 0.087            & 0.080          & 0.040            & 0.043          & 0.007            & 0.005          \\
          &      & 10 & 0.081            & 0.082          & 0.033            & 0.030          & 0.006            & 0.004          &                      &                      & 10                   & 0.070            & 0.069          & 0.038            & 0.038          & 0.007            & 0.006          \\
          &      & 15 & 0.074            & 0.079          & 0.034            & 0.039          & 0.003            & 0.005          &                      &                      & 15                   & 0.070            & 0.078          & 0.022            & 0.027          & 0.003            & 0.002          \\
          & 250  & 1  & 0.118            & 0.118          & 0.053            & 0.056          & 0.008            & 0.010          &                      & 250                  & 1                    & 0.097            & 0.096          & 0.046            & 0.049          & 0.010            & 0.011          \\
          &      & 5  & 0.092            & 0.096          & 0.049            & 0.049          & 0.010            & 0.013          &                      &                      & 5                    & 0.110            & 0.104          & 0.057            & 0.059          & 0.011            & 0.015          \\
          &      & 10 & 0.078            & 0.082          & 0.038            & 0.039          & 0.007            & 0.007          &                      &                      & 10                   & 0.066            & 0.076          & 0.030            & 0.029          & 0.004            & 0.004          \\
          &      & 15 & 0.088            & 0.089          & 0.034            & 0.035          & 0.005            & 0.006          &                      &                      & 15                   & 0.094            & 0.097          & 0.050            & 0.050          & 0.002            & 0.002          \\
          & 500  & 1  & 0.114            & 0.123          & 0.059            & 0.055          & 0.005            & 0.008          &                      & 500                  & 1                    & 0.091            & 0.094          & 0.048            & 0.053          & 0.012            & 0.011          \\
          &      & 5  & 0.087            & 0.094          & 0.04             & 0.038          & 0.004            & 0.004          &                      &                      & 5                    & 0.097            & 0.101          & 0.051            & 0.052          & 0.01             & 0.013          \\
          &      & 10 & 0.106            & 0.106          & 0.053            & 0.054          & 0.013            & 0.01           &                      &                      & 10                   & 0.086            & 0.088          & 0.036            & 0.037          & 0.005            & 0.003          \\
          &      & 15 & 0.100            & 0.099          & 0.058            & 0.054          & 0.011            & 0.012          &                      &                      & 15                   & 0.101            & 0.105          & 0.048            & 0.044          & 0.011            & 0.007          \\
          & 1000 & 1  & 0.104            & 0.108          & 0.043            & 0.045          & 0.010            & 0.011          &                      & 1000                 & 1                    & 0.115            & 0.114          & 0.056            & 0.054          & 0.014            & 0.016          \\
          &      & 10 & 0.105            & 0.107          & 0.059            & 0.060          & 0.008            & 0.008          &                      &                      & 10                   & 0.094            & 0.094          & 0.044            & 0.041          & 0.010            & 0.009          \\
          &      & 20 & 0.104            & 0.107          & 0.055            & 0.057          & 0.009            & 0.010          &                      &                      & 20                   & 0.098            & 0.099          & 0.046            & 0.045          & 0.009            & 0.008          \\
          &      & 30 & 0.112            & 0.104          & 0.056            & 0.058          & 0.015            & 0.014          &                      &                      & 30                   & 0.091            & 0.086          & 0.042            & 0.041          & 0.008            & 0.01           \\
          & 2000 & 1  & 0.098            & 0.096          & 0.043            & 0.045          & 0.004            & 0.005          &                      & 2000                 & 1                    & 0.107            & 0.113          & 0.066            & 0.063          & 0.016            & 0.016          \\
          &      & 10 & 0.096            & 0.096          & 0.058            & 0.058          & 0.009            & 0.009          &                      &                      & 10                   & 0.084            & 0.084          & 0.045            & 0.043          & 0.007            & 0.007          \\
          &      & 20 & 0.098            & 0.098          & 0.056            & 0.057          & 0.013            & 0.012          &                      &                      & 20                   & 0.093            & 0.093          & 0.045            & 0.046          & 0.006            & 0.005          \\
          &      & 30 & 0.107            & 0.107          & 0.053            & 0.055          & 0.009            & 0.008          &                      &                      & 30                   & 0.006            & 0.098          & 0.041            & 0.044          & 0.006            & 0.008          \\ \midrule
F         & 100  & 1  & 0.86             &0.108           & 0.051            & 0.055          & 0.012            & 0.016          & BS                   & 100                  & 1                    & 0.095            & 0.102          & 0.05             & 0.062          & 0.01             & 0.013          \\
          &      & 5  & 0.089            & 0.102          & 0.052            & 0.053          & 0.007            & 0.009          &                      &                      & 5                    & 0.08             & 0.081          & 0.036            & 0.029          & 0.005            & 0.005          \\
          &      & 10 & 0.07             & 0.069          & 0.03             & 0.032          & 0.007            & 0.010          &                      &                      & 10                   & 0.075            & 0.075          & 0.029            & 0.036          & 0.003            & 0.004          \\
          &      & 15 & 0.062            & 0.074          & 0.028            & 0.026          & 0.003            & 0.003          &                      &                      & 15                   & 0.07             & 0.077          & 0.026            & 0.033          & 0.006            & 0.009          \\
          & 250  & 1  & 0.09             & 0.101          & 0.052            & 0.052          & 0.005            & 0.006          &                      & 250                  & 1                    & 0.099            & 0.095          & 0.044            & 0.05           & 0.005            & 0.008          \\
          &      & 5  & 0.074            & 0.08           & 0.031            & 0.032          & 0.007            & 0.008          &                      &                      & 5                    & 0.084            & 0.086          & 0.047            & 0.046          & 0.012            & 0.012          \\
          &      & 10 & 0.088            & 0.083          & 0.038            & 0.04           & 0.003            & 0.003          &                      &                      & 10                   & 0.103            & 0.107          & 0.05             & 0.056          & 0.012            & 0.01           \\
          &      & 15 & 0.078            & 0.087          & 0.039            & 0.038          & 0.006            & 0.007          &                      &                      & 15                   & 0.08             & 0.083          & 0.044            & 0.042          & 0.006            & 0.005          \\
          & 500  & 1  & 0.099            & 0.096          & 0.048            & 0.046          & 0.012            & 0.01           &                      & 500                  & 1                    & 0.12             & 0.117          & 0.057            & 0.055          & 0.008            & 0.01           \\
          &      & 5  & 0.102            & 0.097          & 0.051            & 0.048          & 0.01             & 0.012          &                      &                      & 5                    & 0.105            & 0.113          & 0.051            & 0.052          & 0.012            & 0.012          \\
          &      & 10 & 0.088            & 0.093          & 0.045            & 0.04           & 0.008            & 0.01           &                      &                      & 10                   & 0.101            & 0.116          & 0.043            & 0.043          & 0.007            & 0.008          \\
          &      & 15 & 0.093            & 0.092          & 0.041            & 0.043          & 0.009            & 0.005          &                      &                      & 15                   & 0.101            & 0.101          & 0.036            & 0.042          & 0.003            & 0.005          \\
          & 1000 & 1  & 0.117            & 0.113          & 0.062            & 0.062          & 0.019            & 0.019          &                      & 1000                 & 1                    & 0.1              & 0.098          & 0.047            & 0.047          & 0.011            & 0.011          \\
          &      & 10 & 0.097            & 0.09           & 0.045            & 0.046          & 0.011            & 0.01           &                      &                      & 10                   & 0.09             & 0.088          & 0.035            & 0.038          & 0.005            & 0.008          \\
          &      & 20 & 0.110            & 0.109          & 0.054            & 0.051          & 0.01             & 0.012          &                      &                      & 20                   & 0.097            & 0.101          & 0.052            & 0.05           & 0.013            & 0.013          \\
          &      & 30 & 0.085            & 0.081          & 0.047            & 0.045          & 0.007            & 0.006          &                      &                      & 30                   & 0.096            & 0.096          & 0.039            & 0.039          & 0.006            & 0.007          \\
          & 2000 & 1  & 0.115            & 0.11           & 0.071            & 0.07           & 0.016            & 0.016          &                      & 2000                 & 1                    & 0.098            & 0.103          & 0.044            & 0.049          & 0.013            & 0.011          \\
          &      & 10 & 0.104            & 0.108          & 0.063            & 0.06           & 0.011            & 0.014          &                      &                      & 10                   & 0.087            & 0.087          & 0.043            & 0.046          & 0.011            & 0.01           \\
          &      & 20 & 0.1              & 0.095          & 0.047            & 0.046          & 0.004            & 0.005          &                      &                      & 20                   & 0.08             & 0.075          & 0.035            & 0.037          & 0.01             & 0.009          \\
          &      & 30 & 0.1              & 0.099          & 0.047            & 0.046          & 0.006            & 0.004          &                      &                      & 30                   & 0.101            & 0.103          & 0.048            & 0.046          & 0.01             & 0.011          \\ \bottomrule
\end{tabular}
\end{adjustbox}
\end{table}

\subsection{Functional time series models}\label{sec-fts-model}
We now turn our attention to some examples of serially dependent functional time series.  We consider several models here, the first two in this section for the purpose of performing simulation experiments, and the second two presented in Section \ref{sec-applicatipn} are used for the purpose of time series modelling and residual analyses in two real data examples.  

The time series models that we considered in our simulation experiments are from the general class of functional autoregressive models of order $p$ (FAR($p$)), which take the form 
\[
	X_i(t) = \sum_{j=1}^p\Phi_j(X_{i-p})(t) + \epsilon_i(t) = \sum_{j=1}^p\int \phi_j(t,s)X_{i-j}(s)ds + \epsilon_i(t), 
\]
where $\phi_j(t,s)$ are the kernels of  the integral operators $\Phi_j$, and $\epsilon_i$ is a sequence of iid mean zero functional processes. In our examples, we use the Gaussian kernel
\[
\phi(t,s) = c \exp\left(-\frac{t^2+s^2}{2}\right),
\]
where the constant $c$ is chosen in order that $\norm{\phi} =  |S|\in[0,1]$, where $S= \operatorname{sign}(c)\norm{\phi}$. As such the parameter $S$ describes both the magnitude and ``sign" of the kernel $\phi$. We considered the processes
\begin{itemize}
	\item FAR(1, $S$): $X_i(t) =\int \phi_1(t, s) X_{i-1}(s) d s + \varepsilon_i(t),$  where $\varepsilon_i(t)$ follows the Brownian bridge WN process (BB).
	\item FAR(2, $S_1$, $S_2$): $X_i(t) =  \int \phi_1(t, s) X_{i-1}(s) d s + \int \phi_2(t, s) X_{i-2}(s) d s+\varepsilon_i(t)$, where again $\varepsilon_i(t)$ follows the Brownian bridge WN process (BB).
\end{itemize}

\begin{figure}[ht!]
	\centering
	\begin{subfigure}[b]{0.235\textwidth}
		\centering
		\includegraphics[width=\textwidth]{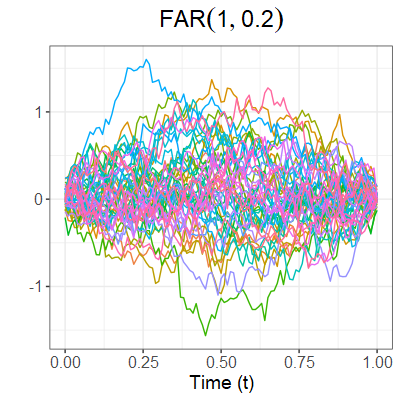}
	\end{subfigure}
	\hfill
	\begin{subfigure}[b]{0.235\textwidth}
		\centering
		\includegraphics[width=\textwidth]{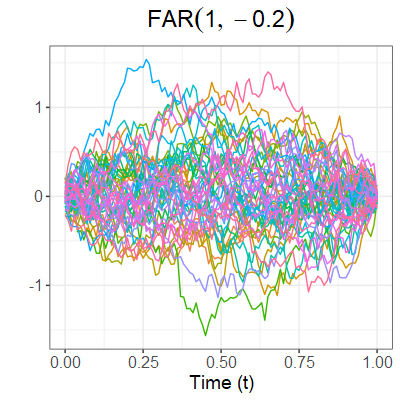}
	\end{subfigure}
	\hfill
	\begin{subfigure}[b]{0.235\textwidth}
		\centering
		\includegraphics[width=\textwidth]{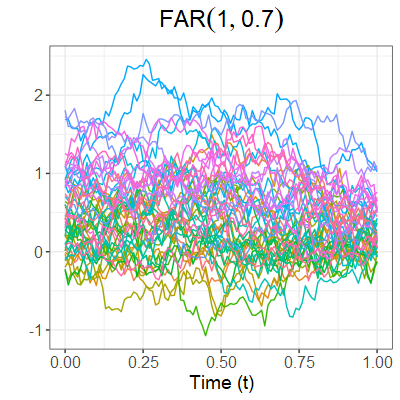}
	\end{subfigure}
	\hfill
	\begin{subfigure}[b]{0.235\textwidth}
		\centering
		\includegraphics[width=\textwidth]{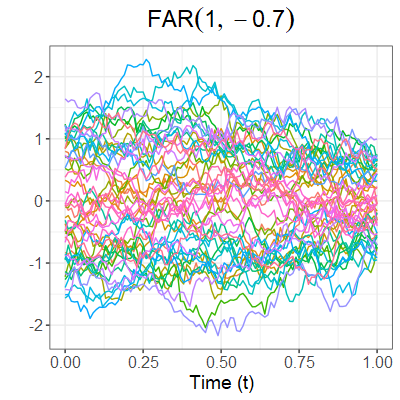}
	\end{subfigure}
	\hfill
	\begin{subfigure}[b]{0.235\textwidth}
		\centering
		\includegraphics[width=\textwidth]{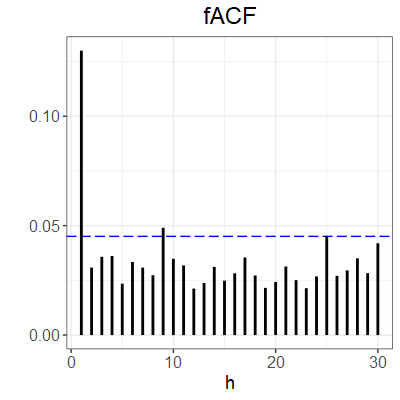}
	\end{subfigure}
	\hfill
	\begin{subfigure}[b]{0.235\textwidth}
		\centering
		\includegraphics[width=\textwidth]{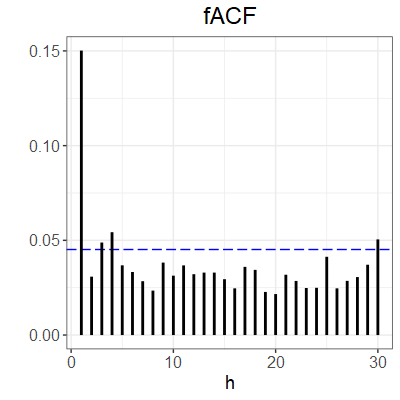}
	\end{subfigure}
	\hfill
	\begin{subfigure}[b]{0.235\textwidth}
		\centering
		\includegraphics[width=\textwidth]{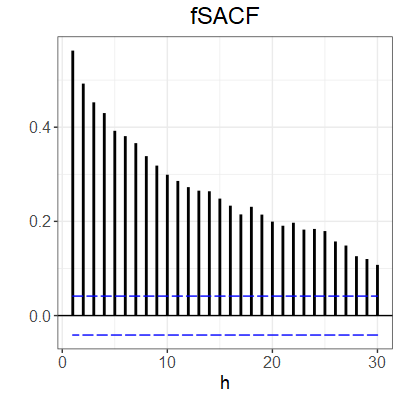}
	\end{subfigure}
	\hfill
	\begin{subfigure}[b]{0.235\textwidth}
		\centering
		\includegraphics[width=\textwidth]{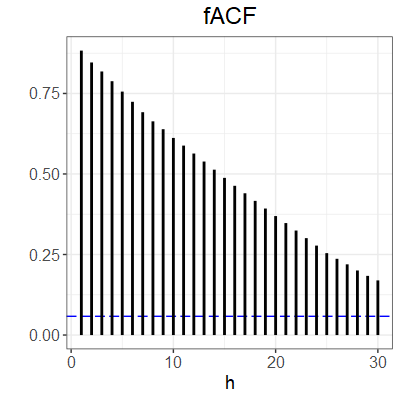}
	\end{subfigure}
	\hfill
	\begin{subfigure}[b]{0.235\textwidth}
		\centering
		\includegraphics[width=\textwidth]{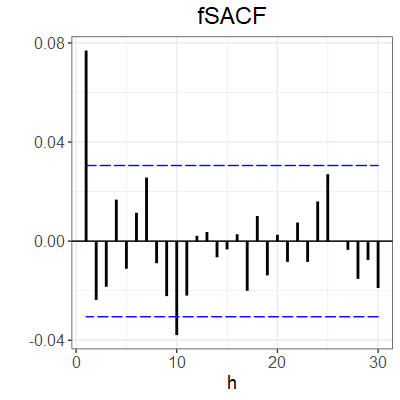}
	\end{subfigure}
	\hfill
	\begin{subfigure}[b]{0.235\textwidth}
		\centering
		\includegraphics[width=\textwidth]{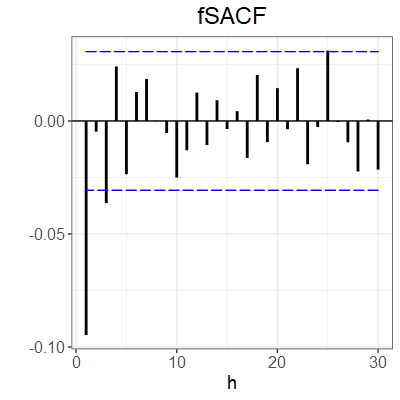}
	\end{subfigure}
	\hfill
	\begin{subfigure}[b]{0.235\textwidth}
		\centering
		\includegraphics[width=\textwidth]{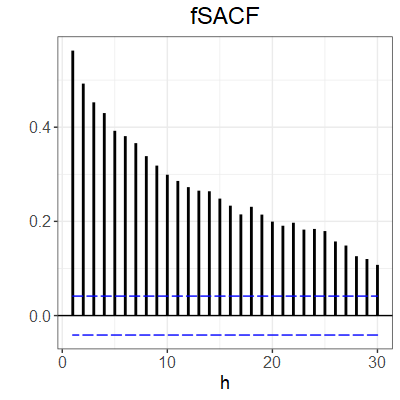}
	\end{subfigure}
	\hfill
	\begin{subfigure}[b]{0.235\textwidth}
		\centering
		\includegraphics[width=\textwidth]{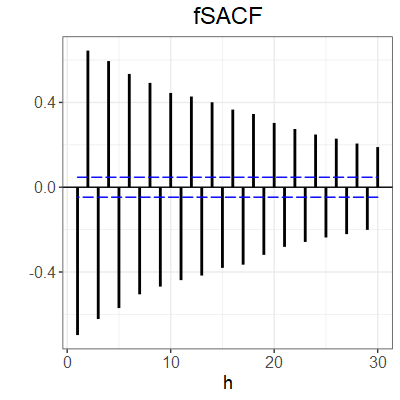}
	\end{subfigure}
	\caption{1000 curves simulated from $\operatorname{FAR}(1,S)$ models with different scaling parameter $S$. \textbf{Top}: 50 iid curves simulated based on each model; \textbf{Middle}: The fACF of each the processes with 95\% confidence bound (in blue)  \textbf{Bottom}: The fSACF of each the processes with 95\% confidence bounds \eqref{conf-def} (in blue).}\label{fig-dependent}
\end{figure}

Figure \ref{fig-dependent} illustrates simulated series of length 1000 from the FAR(1,$S$) process for various values of the parameter $S$, along with corresponding plots of the fACF of \cite{mestre:2021:facf} and the fSACF. This figure highlights one of the advantages of the fSACF compared to the fACF, in that when the parameter $S$ is negative, the  fACF displays a similar behaviour as if the parameter were positive, whereas the fSACF in this case shows an alternating pattern.   

To study the empirical power of the proposed confidence interval and portmanteau test in \eqref{conf-def} and \eqref{port-test}, we applied our tests to 1000 independently generated FAR(1, $S$)  processes of various lengths $n\in \{ 100,250,500,1000\}$ and for increasing  values of $S$. Figure \ref{fig-power} shows the empirical rejection rate of the hypotheses $\mathcal{H}^\prime_{0,H}$ for $H=1,10$ that we observed among the 1000 independent simulations with significance level set at $\alpha=0.05$ in terms of power curves that are a function of $S$. We note that the rejection rate of $\mathcal{H}^\prime_{0,1}$ at level 0.05 is equivalent to the rate at which the confidence interval \eqref{conf-def} does not contain $\hat{\rho}_1$. We observed that the tests appeared to have empirical size close to the nominal size, and that expectedly as $S$ increased to 1, the power of each test also increased to 1. The power was observed to increase more slowly for larger $H$, which we attribute to the fact that the FAR(1,$S$) process has decreasing correlation as a function of the lag.  It was also clear that the empirical power of the test increased as the sample size $n$ increased even when $S$ was small, and the empirical power reached one when the dependence is around $0.45$ regardless of the sample size.

We also performed some simulations to highlight the use of the fSACF in evaluating the goodness-of-fit of functional time series models through residual analysis. Given a model capable of producing fitted/forecasted values of the series $\hat{X}_i(t)$, the residuals are calculated simply as
\begin{equation}\label{eq-fit-resid}
	r_i(t) = X_i(t) - \hat{X}_i(t).
\end{equation}

We considered performing residual analyses of FAR models using an estimator derived from functional principal component analysis. For the FAR(1) model, the kernel $\phi$ may be estimated using a least squares principle by the first $J$ FPCs:
\begin{align}\label{FAR-est}
\hat{\phi}_1(t, s)=\frac{1}{n-1} \sum_{k=1}^{n-1} \sum_{j=1}^J \sum_{i=1}^J \hat{\lambda}_{j}^{-1}\left\langle X_k, \hat{v}_{j}\right\rangle\left\langle X_{k+1}, \hat{v}_{i}\right\rangle \hat{v}_{j}(s) \hat{v}_{i}(t),
\end{align}
where the $\hat{\lambda}_j$ and $\hat{v}_j(\cdot)$ are the estimated eigenvalues and eigenfunctions of the covariance operator of the observations; see e.g. \cite{shang:pca:2014}. Such an estimator may also be constructed for the kernels in general FAR($p$) models, and we provide the details in Appendix \ref{app-FSAR}. The fitted values are then calculated in this case  as 
$$
\hat{X}_i(t) = \int \hat{\phi}_1(t,s)X_{i-1}(s)ds. 
$$

There are a number of ways to determine the number of FPCs $J$ to use, including information criteria
\cite{shibata:1981:selection,yao:2007:selection}, resampling methods \cite{bathia:2010:identifying,hall:vial:2006:assessing} and cross-validation. Here, we use the  cumulative percentage of total variance (CPV)
\[
\operatorname{CPV}(J)=\frac{\sum_{k=1}^J \hat{\lambda}_k}{\sum_{k=1}^n \hat{\lambda}_k},
\]
and choose the smallest $J$ for which $CPV(J)>90\%$. We used the \textbf{fda} package to carry out functional principal component analysis; see \cite{hooker:2009}.   
 
 Figure \ref{fig-misfit} shows three fSACF's: the left most is computed from simulated FAR(2,$S_1$,$S_2$) processes, whereas the middle and right figures are computed from the residuals of an FAR(2) and FAR(1) model fit, respectively. It is clear based on these plots that in each case there remains significant autocorrelation in the residuals based on the FAR(1) model, as expected, while the residuals of the FAR(2) model appear to be reasonably white. 

\begin{figure} [tbp]
	\centering
	\begin{subfigure}[b]{0.45\textwidth}
		\centering
		\includegraphics[width=\textwidth]{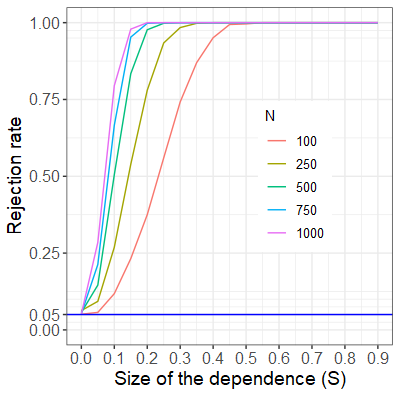}
	\end{subfigure}
	\hfill
	\begin{subfigure}[b]{0.45\textwidth}
		\centering
		\includegraphics[width=\textwidth]{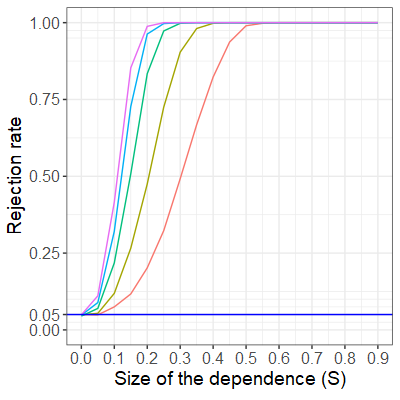}
	\end{subfigure}
\caption{Empirical rejection rate of the portmanteau test under $\hip_{0,H}$ of data generated
according to the FAR(1,S) model based on 1000 independent simulations. \textbf{Left}: H = 1; \textbf{Right}: H = 10. The nominal level $\alpha$ is set to be 0.05.}\label{fig-power}
\end{figure}

\begin{figure} [tbp]
	\centering
	\begin{subfigure}[b]{0.30\textwidth}
		\centering
		\includegraphics[width=\textwidth]{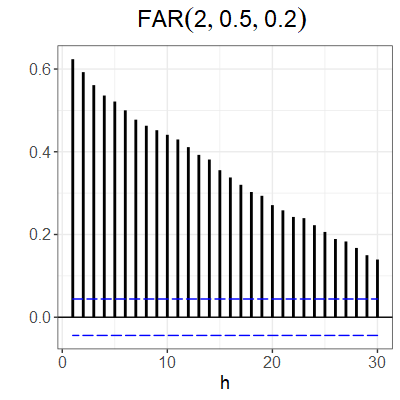}
	\end{subfigure}
	\hfill
	\begin{subfigure}[b]{0.30\textwidth}
		\centering
			\includegraphics[width=\textwidth]{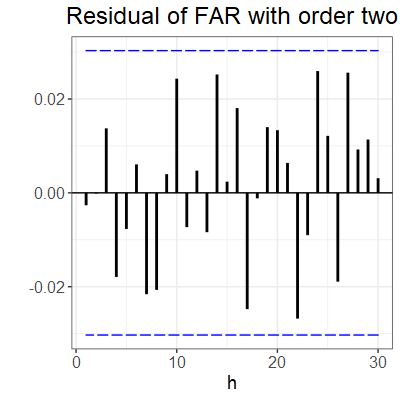}
	\end{subfigure}
	\hfill
	\begin{subfigure}[b]{0.30\textwidth}
		\centering
		\includegraphics[width=\textwidth]{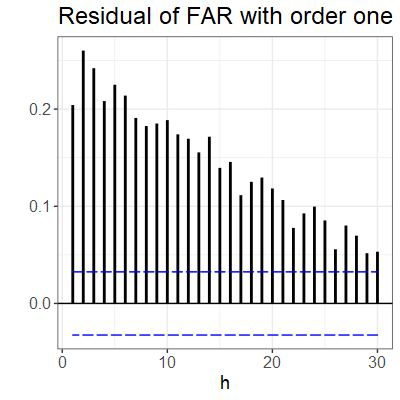}
	\end{subfigure}
	\hfill
	\begin{subfigure}[b]{0.30\textwidth}
		\centering
		\includegraphics[width=\textwidth]{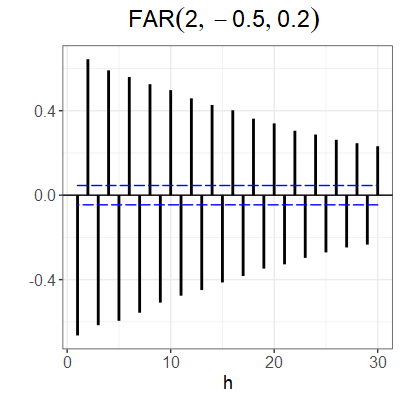}
	\end{subfigure}
	\hfill
	\begin{subfigure}[b]{0.30\textwidth}
		\centering
		\includegraphics[width=\textwidth]{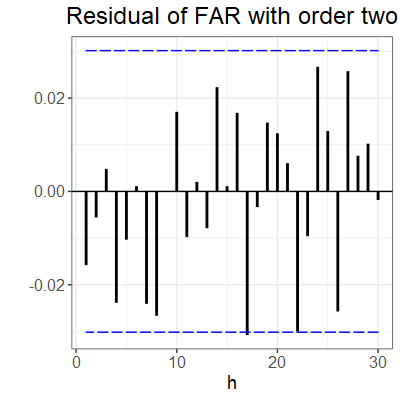}
	\end{subfigure}
	\hfill
	\begin{subfigure}[b]{0.30\textwidth}
		\centering
		\includegraphics[width=\textwidth]{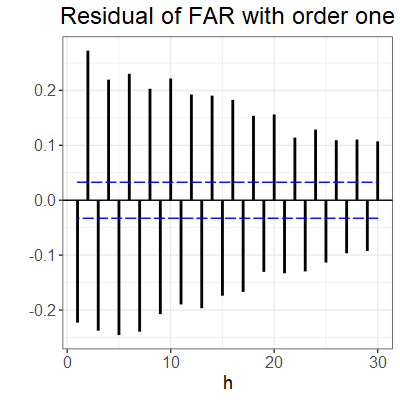}
	\end{subfigure}
	\hfill
	\caption{1000 simulated curves from $\operatorname{FAR}(2,S_1, S_2)$ models with different scaling parameters $S_1$ and $S_2$ with the error process being a Brownian bridge. \textbf{Left}: The fSACF of observations simulated from the FAR(2, $S_1$, $S_2$) models with different coefficients of dependence $S_1$ and $S_2$. \textbf{Middle}: The fSACF of the residuals $r_i$ from an FAR(2) model. \textbf{Right}: The fSACF of the residuals $r_i$ from an FAR(1) model, which in this case is evidently under specified.}\label{fig-misfit}
\end{figure}

\section{Applications}\label{sec-applicatipn}
In this section, we demonstrate the proposed methodology in two  data applications. The first is concerned with fitting forecasting models and evaluating their goodness-of-fit with Spanish electricity price curves. We now introduce the two forecasting models that we consider.

1) Functional seasonal autoregressive model (FSAR)  \cite{horvath:2020:forward-curve}: The model is given as
\begin{equation*}
	X_{i}(t)= \int\phi_{\ell_1}(t,s)X_{i-\ell_{1}}(s)ds+ \dots+\int \phi_{\ell_p}(t,s)X_{i-\ell_p}(s)ds+\varepsilon_i(t),
\end{equation*}
where $\ell_1,...,\ell_p$ are user selected lags. It is clear that the FSAR models are a special cases of general FAR models. We show in  Appendix \ref{app-FSAR} how estimators for the kernels in this model may be constructed based on functional principal component analysis and a least squares principle akin to \eqref{FAR-est}.

\emergencystretch 3em
2) The forecasting model that we consider is the Hyndman-Ullah method \cite{hyndman:ullah:2007:forecast}, which we abbreviate as the ``HU" model: The HU($J$) is a model that is potentially nonlinear and non-stationary that makes use of a truncated Karhunen–Lo\`{e}ve expansion. It supposes that the underlying processes may be well-approximated by its projection onto the first $J$ eigenfunctions of the covariance operator  
\[
	X_i^{(J)}(t) = \hat{m}(t) + \sum_{j=1}^J {\xi}_{i, j} \hat{v}_{j}(t),
\]
where the $\hat{m}$ is the standard sample mean function, ${\xi}_{i,j}$ is referred to as the principal component score, or simply the score, and the $\hat{v}_{j}(\cdot)$ are the eigenfunctions of the sample covariance operator. 

There are three steps to fit and forecast with HU($J$) models: 
\begin{enumerate}
	\item Each $X_i$ is approximated as $X_i(t) \approx \hat m(t) + \sum_{j=1}^J  {\xi}_{i,j} \hat{v}_j(t)$. $J$ may be selected according to a number of criteria, e.g. general cross validation,  information criterion,  cumulative percentage of total variance, or goodness-of-fit testing.
	\item To each of the FPC series, $ {\xi}_{i,j},i=1,...,n$, $j=1,\dots,J$,  we fit a, potentially non-stationary, scalar time series model. Notably the models can be distinct for different score series. Here, we fit SARIMA models, see e.g. \cite{shumway:stoffer:2017:ts}. These models may be used to compute forecasts and fitted values of the score series $\hat{\xi}_{ij}$.
	\item Forecasts or fitted values are constructed as $\hat{X}_i(t) = \hat m(t) + \sum_{j=1}^J  \hat{{\xi}}_{i,j} \hat{v}_j(t)$, and residual curves may be computed as $r_i(t) = X_i(t) - \hat{X_i}(t) = X_i(t) - \sum_{i=1}^J \hat{\xi}_{ij} \hat{v}_j(t)$.
\end{enumerate}

\subsection{Spanish electricity data}
The first application we consider is Spanish electricity price data observed hourly from January 1st, 2014, to December 31st, 2014. The data was originally provided by the Spanish electricity market operator, and are also available in the R \textbf{fdaACF} package \cite{fdaACF:package}. \cite{gonzalez:2017:ARMAX} studied fitting and forecasting these data using an autoregressive–moving-average model with exogenous inputs (ARMAX), and \cite{mestre:2021:facf} uses the same dataset to illustrate the properties of the fACF. We use $\{Y_i(t),~i=1,2,\dots,365, ~ t\in[1,24]\}$ to denote the linearly interpolated Spanish electricity price curves, and these are displayed in  the top left panel of Figure \ref{fig-electricity}. Although each curve generally evolves according to a standard daily pattern of low prices in the early morning hours, and higher prices during the times of peak demand, we also observe that there is a strong weekly and seasonal trend to the series over the course of the year. The corresponding fACF and fSACF presented in the middle and right of the top row of Figure \ref{fig-electricity} both indicate such a pattern. To remove the trend component, we consider the pointwise differenced series  $D_i(t)=Y_i(t) - Y_{i-1}(t)$. The bottom row of Figure \ref{fig-electricity} displays the differenced series along with the corrsponding fACF and fSACF plots. The differenced series appears roughly mean stationary, and has prominent autocorrelation corresponding to the weekly lag $h=7$. In particular, the fSACF indicates that the correlation at lags one and two are ``negative", while the correlation at weekly lags is ``positive", which is to be expected. In contrast, the fACF does not capture this pattern. 

We then fit FSAR and HU($J$) models to the differenced series $D_i$, and examined their goodness-of-fit using a residual analysis based on the fSACF.  For the FSAR model, we chose lags  $1$ and $7$, as we see the data exhibited both lag $1$ and weekly autocorrelation, and selected the model dimension using the CPV criterion explaining 90\% of the total variation. We also considered HU($J$) models for $J \in \{1,...,5\}$, and to each resulting scalar score series we fit a SARIMA model using the automated model selection method in \cite{hyndman:khandakar:2008}. In particular, this lead to SARIMA models with orders $(1,0,2)\times(2,0,0)_7$ , $(1,0,2)\times(2,0,0)_7$, $(2,0,2)\times(2,0,0)_7$, $(2,0,1)\times(2,0,1)_7$ and $(2,0,2)\times(2,0,0)_7$ for the first five univariate score series, respectively, where $(a,b,c)\times(A,B,C)_s$ describes a SARIMA model with seasonal lag $s$, standard and seasonal differencing $b$ and $B$,   standard and seasonal autoregressive orders $a$ and $A$, and standard and seasonal moving average orders of $c$ and $C$.  Figure \ref{fig-electricity-fit} shows the fSACF of the residuals of these models. We observed that the FSAR model does not fit the data well, as strong autocorrelation is still present in the residuals at low lags as well as in the weekly lag. This did not improve by increasing the CPV value up to 99\%. As for the HU$(J)$ model with $J \in \{3,4,5\}$, most of the fSACF values fall within the 95\% confidence bounds coloured in blue. We additionally computed the P-values of tests of $\mathcal{H}^\prime_{0,H}$ as a function of $H$ based on the statistic  $\hat{Q}_{364,H}$ applied to the residuals of each model in Figure \ref{fig-electricity-fit}. These suggest again that the FSAR fit and HU(1) fit are poor. As for HU(3), these P-values for small lags were slightly above 0.05, while for HU(5)  the P-values were for most lags above the 0.05 threshold. This suggests that, although there may still be some weak, lingering autocorrelation in the residuals of the HU(5) model at the weekly lag, this model generally appears to fit well.  

\begin{figure}[ht!]
	\centering
	\begin{subfigure}[b]{0.3\textwidth}
		\centering
		\includegraphics[width=\textwidth]{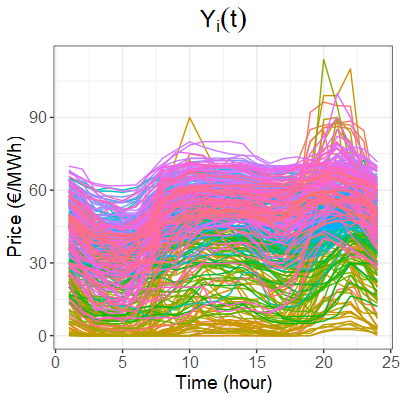}
	\end{subfigure}
	\hfill
	\begin{subfigure}[b]{0.3\textwidth}
		\centering
		\includegraphics[width=\textwidth]{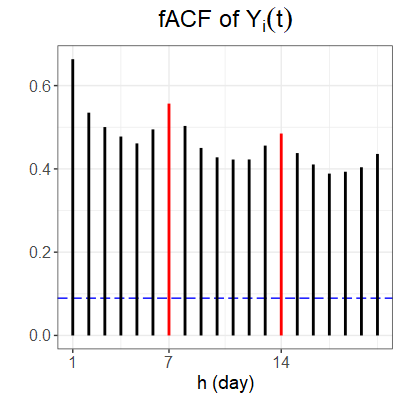}
	\end{subfigure}
	\hfill
	\begin{subfigure}[b]{0.3\textwidth}
		\centering
		\includegraphics[width=\textwidth]{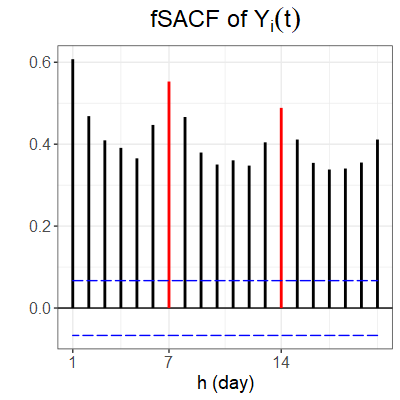}
	\end{subfigure}
	\hfill 
	\begin{subfigure}[b]{0.3\textwidth}
		\centering
		\includegraphics[width=\textwidth]{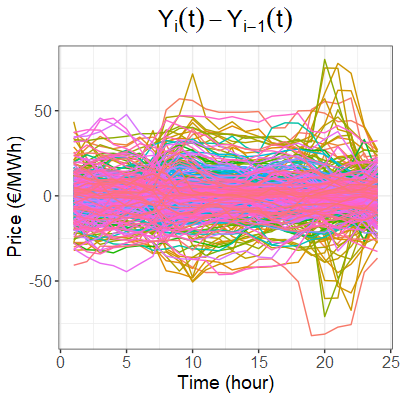}
	\end{subfigure}
	\hfill
	\begin{subfigure}[b]{0.3\textwidth}
		\centering
		\includegraphics[width=\textwidth]{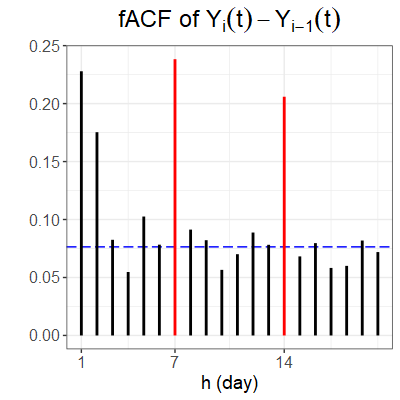}
	\end{subfigure}
	\hfill 
	\begin{subfigure}[b]{0.3\textwidth}
		\centering
		\includegraphics[width=\textwidth]{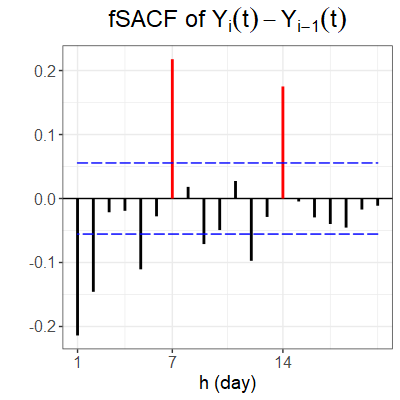}
	\end{subfigure}
	\caption{\textbf{Left}: Top is the hourly observed annual Spanish electricity data $\{Y_i(t),i=1,\dots,365,t=1,\dots,24\}$ from January 1st 2014 to December 31 sgemented by day. Bottom is the differenced curves $\{Y_i(t)-Y_{i-1}(t), t\in[0,24]\}$; \textbf{Middle}: The fACF of $Y_i(t)$ (top) and $Y_i(t)-Y_{i-1}(t)$ (bottom), respectively;  \textbf{Right}: The fSACF of $Y_i(t)$ (top) and $Y_i(t)-Y_{i-1}(t)$ (bottom), respectively. The blue lines in the middle and right panels are the 95\% confidence interval under the null hypothesis $\mathcal{H}_{0,h}$, for $h=1,\dots,20$.}\label{fig-electricity}
\end{figure}

\begin{figure}[ht!]
	\centering
	\begin{subfigure}[b]{0.3\textwidth}
		\centering
		\includegraphics[width=\textwidth]{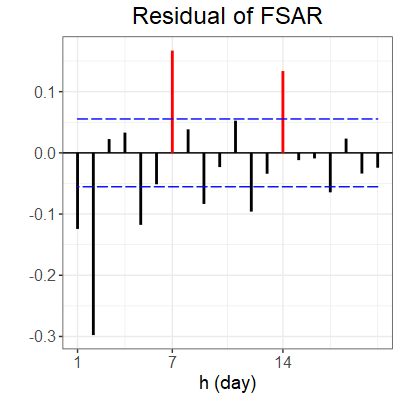}
	\end{subfigure}
	\hfill
	\begin{subfigure}[b]{0.3\textwidth}
		\centering
		\includegraphics[width=\textwidth]{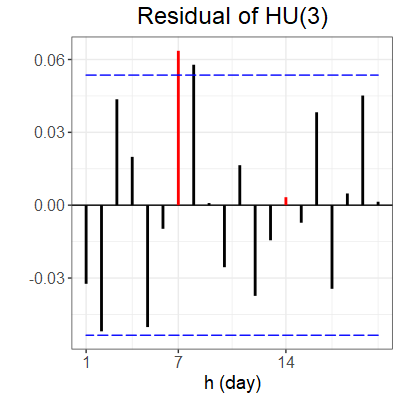}
	\end{subfigure}
	\hfill
	\begin{subfigure}[b]{0.3\textwidth}
		\centering
		\includegraphics[width=\textwidth]{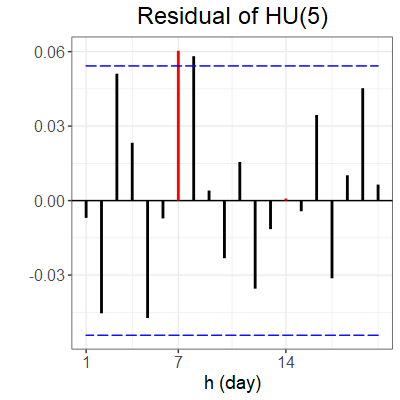}
	\end{subfigure}
	\hfill
	\begin{subfigure}[b]{0.8\textwidth}
		\centering
		\includegraphics[width=\textwidth]{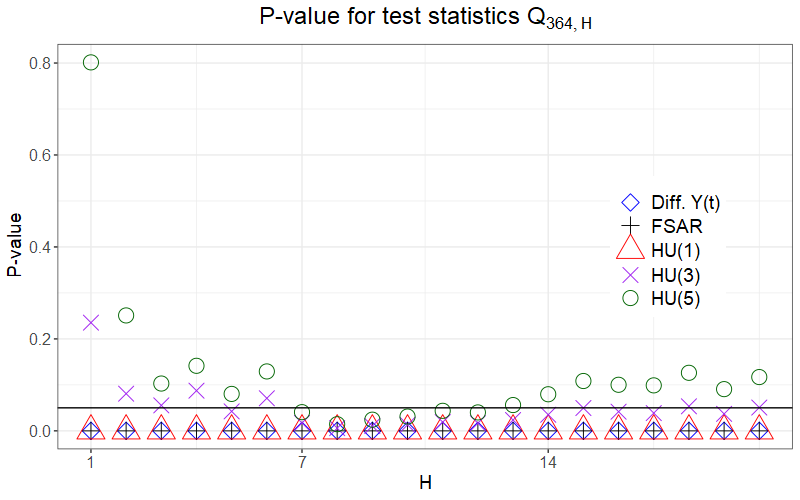}
	\end{subfigure}
	\caption{\textbf{Left}: The fSACF of the residual $r_i$ after fitting by the FSAR model; using the ; \textbf{Middle}: The fSACF of the residual $r_i$ after fitting the data using the Hyndman-Uallh model with three FPCs; \textbf{Right}: The fSACF of the residual $r_i$ after fitting the data using the Hyndman-Uallh model with three FPCs;
		\textbf{Bottom} The P-values of test statistics $\hat{Q}_{364,H}$ under $\mathcal{H}_{0,H}^\prime$ of residual using FSAR model, residual using HU model with three and five FPCs, and the differenced series $Y_i(t) - Y_{i-1}(t)$, respectively.} \label{fig-electricity-fit}
\end{figure}

\subsection{Densely observed intraday asset price data}
The second application that we considered was to the densely observed intraday asset price data used in \cite{kokoszka:2017:autocovariance-ftsa-conditional}. This dataset consists of financial asset prices, including commodity and currency exchange rates, obtained from the Chicago Mercer Exchange ({\textit{cmegroup.com}}), 
S\&P 500 index, and Apple stock price obtained from Nasdaq ({\textit{nasdaq.com}}). Each series that we consider is tabulated in  Table \ref{table:stock}. There are 78 daily observations spanning 249 trading days within the time span from Jan 02, 2014, to Dec 31, 2014.

\begin{table}[tbp]
\centering
\caption{Assets used in the study spanned from Jan 02, 2014, to Dec 31, 2014.}
\label{table:stock}
\begin{tabular}{@{}lll@{}}
	\toprule
	Class             & Notation & Description                   \\ \midrule
	Index             & S\&P 500 & Standard and Poor's 500       \\ 
	Currency exchange & EC       & Euro to US Dollar             \\
	Commodity futures & CL       & Light sweet crude oil futures \\
	Technology        & AAPL     & Apple Inc.                    \\ \bottomrule
\end{tabular}
\end{table}

Let $P_i(t)$ denote the intraday asset price of an asset on the $i$-th day at time $t$, which we obtained from the raw data using linear interpolation. On each trading day, the first return is recorded at 9:35, and the last one is at 16:00. We consider four functional time series derived from these data: 1) five minute log return sequence
\[
R_{i}(t)= \ln P_{i}(t)- \ln  P_{i}(t-5),
\]
2) the cumulative intraday returns (CIDRs)
\[
C_i(t) = \ln P_i(t) - \ln P_i(0),
\]
 3) the squared sequence of five minute log return curve $R_i^2(t) $, and 4) squared CIDR sequence $C_i^2(t)$. The CIDR may be thought of as a normalized, smoothed curve that interpolates the standard log-returns with intraday data, which can potentially have high volatility; see, \cite{ji:2021:CIDR,kokoszka:2019:CIDR,rice:2020:CIDR}.

Figure \ref{fig-stock} displays the S\&P 500 five minute log return $R_i(t)$, squared S\&P 500 log return $R_i^2(t)$, the CIDR $C_i(t)$, and the squared CIDR $C_i^2(t)$ in the first four days, respectively. It appears that both the log return curves $R_i(u)$ and CIDRs $C_i(t)$ are roughly stationary.   We then calculate the fSACF of each series, which is displayed in the case of light sweet crude oil futures (CL) in the the bottom row of Figure \ref{fig-stock}. We observed that  both the five minute log-return series and CIDR series appear based on their fSACFs to be white noise processes.  As for the squared curves, we observed that most series exhibited significant, although weak,  spherical autocorrelation, especially in early lags. The cumulative significance of the observed spherical autocorrelations was assessed using the portmanteau test of $\mathcal{H}^\prime_{0,10}$, the P-values of which are reported in Table \ref{tab-port}. The appearance of white noise of the original series coexisting with correlation in the squared series is the hallmark of GARCH type behaviour or volatility. Confirming this observation with the fSACF, which is more robust to outliers and lower order moments in the data generating process, helps confirm that the patterns in the autocorrealtion observed previously in these series is not simply the consequence of outliers.   

\begin{table}[tbp]
	\centering
	\caption{Value of the test statistics $\hat{Q}_{249,10}$ and associated P-values computed from \eqref{port-test} for: 1) five minute log returns $R_i(t)$, 2) squared five minute log returns $R_i^2(t)$, 3) CIDR $C_i(t)$, and 4) squared CIDR $C_i^2(t)$, for each asset considered.}
	\label{tab-port}
	\begin{adjustbox}{max width=\textwidth}
	\begin{tabular}{@{}llllllllll@{}}
		\toprule
		\multicolumn{2}{c}{Type} & \multicolumn{2}{c}{$R_i(t)$} & \multicolumn{2}{c}{$R_i^2(t)$} & \multicolumn{2}{c}{$C_i(t)$} & \multicolumn{2}{c}{$C_i^2(t)$} \\ \midrule
		&& $\hat{Q}_{249,10}$     & P-value     & $\hat{Q}_{249,10}$      & P-value      & $\hat{Q}_{249,10}$     & P-value     & $\hat{Q}_{249,10}$      & P-value      \\ \midrule
		S\&P 500 &                  & 3.0144         & 0.9811      & 31.0026         & 0.0006        & 6.2736         & 0.7918      & 12.9304         & 0.2276       \\
		EC &                        & 7.2702         & 0.6997      & 227.4773        & 0.0000            & 4.4191         & 0.9265      & 166.4873        & 0.0000            \\
		CL   &                      & 6.6572         & 0.7574      & 390.0302        & 0.0000            & 4.4574         & 0.9244      & 210.533         & 0.0000            \\
		AAPL  &                    & 5.8559         & 0.8272      & 24.0211         & 0.0080       & 6.8985         & 0.7350      & 7.7158          & 0.6566       \\ \bottomrule
	\end{tabular}
	\end{adjustbox}
\end{table}

\begin{figure}[tbp]
	\centering
	\begin{subfigure}[b]{0.235\textwidth}
		\centering
		\includegraphics[width=\textwidth]{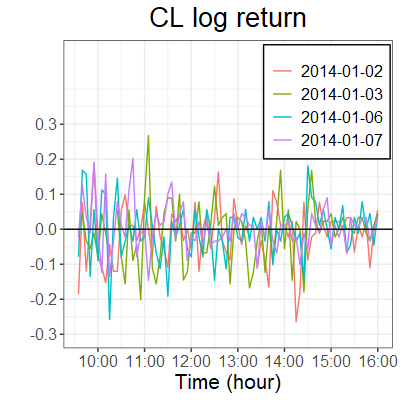}
	\end{subfigure}
	\hfill
	\begin{subfigure}[b]{0.235\textwidth}
		\centering
		\includegraphics[width=\textwidth]{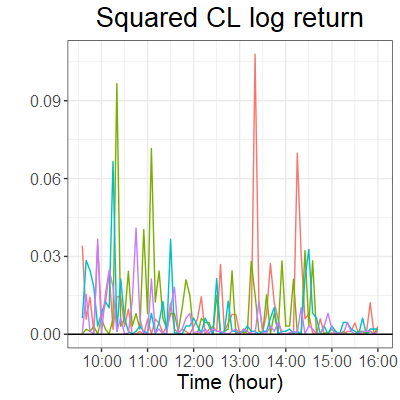}
	\end{subfigure}
	\hfill
	\begin{subfigure}[b]{0.235\textwidth}
		\centering
		\includegraphics[width=\textwidth]{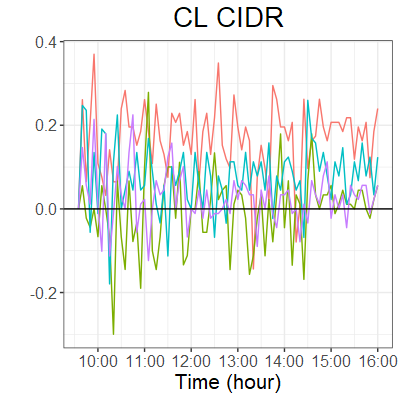}
	\end{subfigure}
	\hfill
	\begin{subfigure}[b]{0.235\textwidth}
		\centering
		\includegraphics[width=\textwidth]{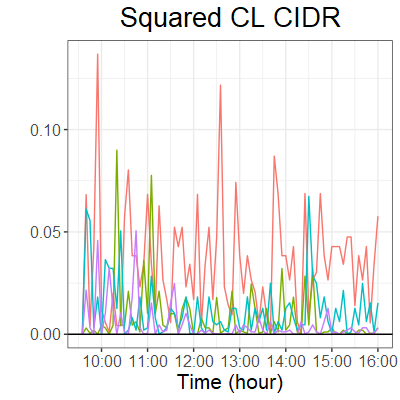}
	\end{subfigure}
	\hfill
	\begin{subfigure}[b]{0.235\textwidth}
		\centering
		\includegraphics[width=\textwidth]{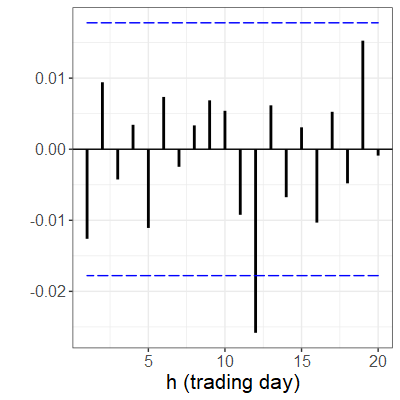}
	\end{subfigure}
	\hfill
	\begin{subfigure}[b]{0.235\textwidth}
		\centering
		\includegraphics[width=\textwidth]{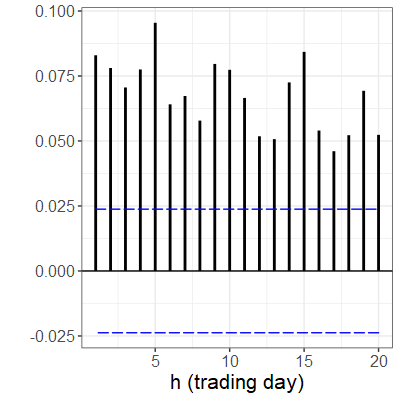}
	\end{subfigure}
	\hfill
	\begin{subfigure}[b]{0.235\textwidth}
		\centering
		\includegraphics[width=\textwidth]{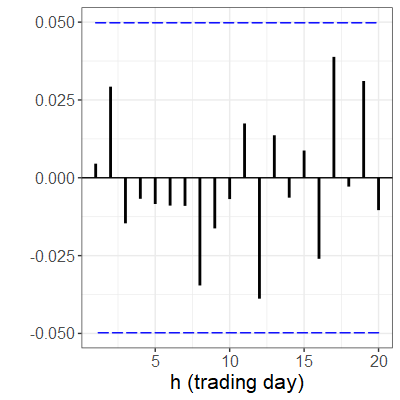}
	\end{subfigure}
	\hfill
	\begin{subfigure}[b]{0.235\textwidth}
		\centering
		\includegraphics[width=\textwidth]{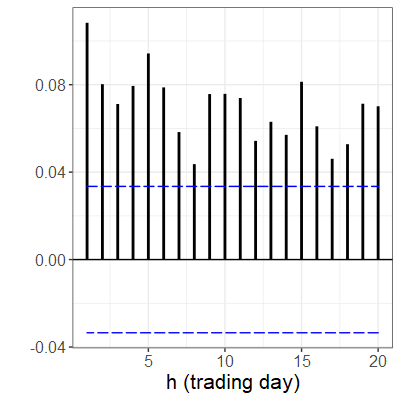}
	\end{subfigure}
	\caption{\textbf{Top}: samples of four curves from each of the four functional time series considered, constructed from the intraday light sweet crude oil futures prices (CL). \textbf{Bottom}: the empirical fSACF of each to these series along with the 95\% confidence interval \ref{conf-def}.  }\label{fig-stock}
\end{figure}

\section{Conclusion}\label{sec-conclusion}
We have put forward a new autocorrelation measure for functional time series: the functional spherical autocorrelation function (fSACF). Unlike existing methods in the literature that rely on measuring the size of the empirical autocovariance operators, the approach we propose is based on measuring the inner product between  projections of lagged pairs of the series onto the unit sphere. Some advantages conferred by this approach are its ability to capture the direction or ``sign" of the serial dependence in the series at a given lag, and it is more robust to outliers/lower order moments of the data generating process. 

A central limit theorem for the fSACF  was established under the assumption that the underlying process is a strong white noise, and mild additional conditions were also put forward under which this result remains true when the spatial median defining the fSACF is replaced with an estimator. These large sample results were used to derive confidence intervals and portmanteau tests for the fSACF. Through Monte Carlo simulations, we showed that these confidence intervals and tests work well in finite samples, highlighting that they also appear to be robust against outliers/low moments in the data generating mechanism, and can be useful in performing model selection with functional time series models. Two real data applications illustrated the use of the fSACF for model selection, and also helped confirm with a more robust approach the finding that intraday return curves derived from dense asset prices series appear to evolve as volatile white noise series. 

\setcounter{equation}{0}
\setcounter{figure}{0}
\setcounter{table}{0}
\makeatletter
\renewcommand{\theequation}{S\arabic{equation}}
\renewcommand{\thefigure}{S\arabic{figure}}
\renewcommand{\thetable}{S\arabic{table}}

\setcounter{section}{0}
\renewcommand{\thesection}{A.\arabic{section}}
\begin{appendices}
\section{Technical details and roofs}
Throughout these proofs, we assume without loss of generality  that the spatial median $\mu=0$. We use $C_i$, $i\ge 1$ to  denote unimportant positive numerical constants that may change from line to line.  

\subsection{Proof of Theorem 1}\label{app-thm1}
\begin{proof} We use the notation 
\begin{equation} 
	\sqrt{n}\tilde{\rho}_h = \frac{1}{\sqrt{n} }\sum_{i=1}^{n-h} \inp{S(X_i)}{S(X_{i+h})}=: \frac{1}{\sqrt{n} }\sum_{i=1}^{n-h} T_{i,h}.
\end{equation}
$\{T_{i,h}\}_{i \in \mathbb{Z}}$ forms a strictly stationary and $h-$dependent sequence of scalar random variables. By equation (2) of \cite{gervini:2008:robust-fpca} and the independence of the $X_i$'s, $\E T_{i,h}=0$. Using Fubini's theorem and the independence of $X_i$'s, we obtain that 
\begin{align*}
\E T_{i,h}^2 &= \E \inp{S(X_i )}{S(X_{i+h} )}^2 \\
&=  \int \hspace{-.3cm}\int \E S(X_i )(t)S(X_{i+h} )(t)S(X_i )(s)S(X_{i+h} )(s)dtds =\|C_P\|^2_2.
\end{align*}
Once again by equation (2) of \cite{gervini:2008:robust-fpca} and the independence of the $X_i$s, we obtain that for $i \ne j$, $\E T_{i,h}T_{j,h}=0$. Hence by Theorem 6.4.2 of \cite{brockwell:davis:1991:ts}, $\sqrt{n}\tilde{\rho}\stackrel{D}{\to} N(0,\|C_P\|^2_2).$ It follows by elementary arguments that for fixed $H$, $\sqrt{n}\|\tilde{R}_H  - {R}_H\|_E \stackrel{P}{\to}0,$ where 
$$
R_H = \frac{1}{n} \sum_{i=1}^n (T_{i,1},...,T_{i,H})^\top =: \frac{1}{n} \sum_{i=1}^n \bm{T}_{i,H}. 
$$
Similarly as above, the vector valued process $\{ \bm{T}_{i,H} \}_{i\in \mathbb{Z}} $ is strictly stationary, mean zero, and $H$ dependent, with $ \E \bm{T}_{i,H} \bm{T}_{i,H}^\top  = \|C_P\|^2 \mathbb{I}_H $, and  $ \E \bm{T}_{i,H} \bm{T}_{j,H}^\top  =0 $, if $i \ne j$. Hence by Proposition 11.2.2 of \cite{brockwell:davis:1991:ts}, $ \sqrt{n}\tilde{R}_H \stackrel{D}{\to} N_H(0,\|C_P\|^2 \mathbb{I}_H)$. It follows from the continuous mapping theorem that $\|\sqrt{n}\tilde{R}_H\|_E^2   \stackrel{D}{\to}  \|C_P\|^2  \chi^2(H).$ 
\end{proof}

\subsection{Proof of Theorem \ref{thm-est-white-con}}

\begin{lemma}\label{lem-sqnorm}
  Under Assumption \ref{as-moment-2}, if $\|u\| \le C_1$, then there exists a constant $C_3$ so that

  $$
  \E \frac{\|X_i\|^2}{\|X_i-u \|^2} \le C_3.
  $$
\end{lemma}
\begin{proof} For all $M>0$,
  \begin{align*}
  \E \frac{\|X_i\|^2}{\|X_i-u \|^2}  &=   \E \frac{\|X_i\|^2}{\|X_i-u \|^2} \mathds{1}\{ \|X_i-u \| > M\}+ \E \frac{\|X_i\|^2}{\|X_i-u \|^2} \mathds{1}\{ \|X_i-u \| \le M\} \\
  &\le \E \frac{\|X_i\|^2}{\|X_i-u \|^2} \mathds{1}\{ \|X_i-u \| > M\} + [M+C_1]C_2,
  \end{align*}
where in the second line we applied Assumption \ref{as-moment-2}, and used that $\|X_i-u \| \le M$ implies by the triangle inequality that $\|X_i\| \le M+C_1$.  Again by the triangle inequality $\|X_i\| - C_1 \le \|X_i - u\|$, which implies that so long as $\|X_i\|> C_1$,
$$\frac{\|X_i\|}{\|X_i-u \|} \le \frac{\|X_i\|}{\|X_i\| - C_1}.$$
Since $\|X_i - u \| \le \|X_i\| + C_1$, we have if $M> 2C_1$, then $ \|X_i-u \| > M$ implies that $\|X_i\|>C_1$. This implies then with $M> 2C_1$ that
$$\frac{\|X_i\|^2}{[\|X_i\| - C_1]^2} \mathds{1}\{ \|X_i-u \| > M\}  \le C_5$$
almost surely for a positive constant $C_5$. Hence Lemma \ref{lem-sqnorm} holds with $C_3 = C_5 + [M+C_1]C_2 $ for any $M> 2C_1$.
\end{proof}

Theorem \ref{thm-est-white-con} follows immediately from Theorem \ref{thm-1} and the following additional lemma. 
\begin{lemma} Under Assumptions \ref{as-tight} and \ref{as-moment-2}, $\sqrt{n} (\hat{\rho}_h - \tilde{\rho}_h) = o_P(1)$.
\end{lemma}




\begin{proof}
Letting $\hat{m}_{i,h} = \|X_i\| \|X_{i+h}\| \|X_i - \hat{\mu}\| \| X_{i+h} - \hat{\mu}\|$, we have that

\begin{align*}
  &\sqrt{n} (\hat{\rho}_h  - \tilde{\rho}_h)  \\
  &= \frac{1}{\sqrt{n}} \sum_{i=1}^{n} \left[\inp{S(X_i-\hat{\mu})}{S(X_{i+h}-\hat{\mu})}-\inp{S(X_i)}{S(X_{i+h})}  \right] \\
  &= \frac{1}{\sqrt{n}} \sum_{i=1}^{n} \frac{1}{\hat{m}_{i,h}} \left[ \langle (X_i - \hat{\mu})\|X_i\|,(X_{i+h}- \hat{\mu})\|X_{i+h}\| \rangle   \right. \\
  &\left. \;\;\;
  -\langle X_i\|X_i- \hat{\mu}\|, X_{i+h}\|X_{i+h}- \hat{\mu}\| \rangle   \right] \\
  &= \frac{1}{\sqrt{n}} \sum_{i=1}^{n} \frac{1}{\hat{m}_{i,h}} \left[ \langle  X_i  \|X_i\|, X_{i+h}  \|X_{i+h}\| \rangle - \langle X_i\|X_i - \hat{\mu}\|, X_{i+h}\|X_{i+h}- \hat{\mu}\| \rangle   \right] \\
  &\;\;\; - \frac{1}{\sqrt{n}} \sum_{i=1}^{n} \frac{\langle  \hat{\mu} \|X_i\|, X_{i+h} \|X_{i+h}\| \rangle}{\hat{m}_{i,h}}-  \frac{1}{\sqrt{n}} \sum_{i=1}^{n} \frac{\langle  X_i\|X_i\|, \hat{\mu} \|X_{i+h}\| \rangle}{\hat{m}_{i,h}}  \\
  &\;\;\; +  \frac{1}{\sqrt{n}} \sum_{i=1}^{n} \frac{\langle  \hat{\mu}\|X_i\|, \hat{\mu} \|X_{i+h}\| \rangle}{\hat{m}_{i,h}} \\
  &=: R_{1,n} - R_{2,n}-R_{3,n} + R_{4,n}.
\end{align*}
We now aim to show that $R_{j,n} = o_P(1)$, $j=1,...,4$. Starting with $R_{2,n}$, we have by the linearity of the inner product and the Cauchy-Schwarz inequality that

\begin{align}\label{s4}
| R_{2,n}| &= \left| \frac{1}{\sqrt{n}} \sum_{i=1}^{n} \frac{\langle  \hat{\mu} \|X_i\|, X_{i+h} \|X_{i+h}\| \rangle}{\hat{m}_{i,h}} \right|   =     \left| \frac{1}{\sqrt{n}} \sum_{i=1}^{n} \frac{\langle  \hat{\mu}  , X_{i+h} \rangle}{\|X_i - \hat{\mu}\|\|X_{i+h} - \hat{\mu}\|} \right|             \\
&= \left|  \left\langle \hat{\mu}  , \frac{1}{\sqrt{n}} \sum_{i=1}^{n} \frac{X_{i+h}}{\|X_i - \hat{\mu}\|\|X_{i+h} - \hat{\mu}\|} \right\rangle \right|  \notag \\
&\le \| \sqrt{n}\hat{\mu}\|  \left\| \frac{1}{n} \sum_{i=1}^{n} \frac{X_{i+h}}{\|X_i - \hat{\mu}\|\|X_{i+h} - \hat{\mu}\|} \right\| =: \| \sqrt{n}\hat{\mu}\| \|R^*_{2,n}\|.
 \end{align}
By Assumption \ref{as-tight}, $\| \sqrt{n}\hat{\mu}\|= O_P(1)$. Let $\epsilon,\delta >0$. Also by Assumption \ref{as-tight}, there exists a compact set $K\subset \mathcal{H}$ so that

\begin{align}\label{s5}
P(\|R^*_{2,n}\| > \epsilon) \le \frac{\delta}{2} + P(\{\|R^*_{2,n}\| > \epsilon \} \cap \{ \sqrt{n}\hat{\mu} \in K\}).
\end{align}

Define the process $G_{2,n}: \mathcal{H} \to \mathcal{H}$ by

$$
G_{2,n}(u)= \frac{1}{n} \sum_{i=1}^{n} \frac{X_{i+h}}{\|X_i - n^{-1/2}u\|\|X_{i+h} - n^{-1/2}u\|},
$$
so that $R^*_{2,n} = G_{2,n}(\sqrt{n}\hat{\mu}).$ Then
\begin{align}\label{s-7}
P(\{\|R^*_{2,n}\| > \epsilon \} \cap \{ \sqrt{n}\hat{\mu} \in K\}) \le P\left( \sup_{u \in K} \|G_{2,n}(u)\| > \epsilon \right).
\end{align}
Using that $\{X_i\}_{i\in \mathbb{Z}}$ are iid, we get that
$$
\E G_{2,n}(u) = \E \|X_0 - n^{-1/2}u\|^{-1} \E \frac{X_{ h}}{\|X_{h}- n^{-1/2}u\|}.
$$

Notice that for each $u \in K$,

$$Y_{h,n}:= \left\| \frac{X_{ h}}{\|X_{h}- n^{-1/2}u\|}- \frac{X_{ h}}{\|X_{h}\|} \right\| \stackrel{a.s.}{\to} 0 \mbox{ as } n\to \infty.$$
Moreover, it follows from Lemma \ref{lem-sqnorm} that $\{Y_{h,n}\}_{n\ge 1}$ are uniformly integrable. Hence by the Vitali convergence theorem (see e.g. Theorem 16.14 of Billingsly (1995)), $\E Y_{h,n} \to 0$ as $n \to \infty$. According to the contraction property of expectation in Hilbert space (see pg. 29 of Bosq (2000)), it follows that
\begin{align}\label{mean-conv-1}
\left\| \E \frac{ X_{ h}}{\|X_{h}- n^{-1/2}u\|}  \right\| \to 0 \mbox{ as }n\to \infty.
\end{align}

 Since $\E \|X_0 - n^{-1/2}u\|^{-1} < \infty$ by Assumption \ref{as-moment-2}, we then obtain that $\E G_{2,n}(u) \to 0$ as $n\to \infty$ for each $u \in K$. Also for $u,u' \in K$, we obtain using simple arithmetic with the reverse triangle inequality  and again Assumption \ref{as-moment-2} that
\begin{align*}
&\|\E G_{2,n}(u) -\E G_{2,n}(u')\|  \\
&=  \E \frac{\|X_0 - n^{-1/2}u\| -\|X_0 - n^{-1/2}u'\| }{ \|X_0 - n^{-1/2}u\| \|X_0 - n^{-1/2}u'\|} \left\| \E \frac{X_h}{\|X_h-n^{-1/2}u\|} \right\| \\
&\;\;\;\; + \left\| \E \frac{X_h[\|X_h - n^{-1/2}u\| -\|X_h - n^{-1/2}u'\|]}{\|X_h- n^{-1/2}u\| \|X_h - n^{-1/2}u'\|} \right\| \E \|X_0 - n^{-1/2}u'\|^{-1} \\
&=C_6\|u-u'\|.
\end{align*}
Hence $\E G_{2,n}(u)$ is uniformly continuous, from which it follows in conjunction with  $\E G_{2,n}(u)=o(1)$ that
\begin{align}\label{s-2}
  \sup_{u \in K} \|\E G_{2,n}(u)\| =o(1), \mbox{ as } n\to \infty.
\end{align}
Also,   using that $\{X_i\}_{i\in \mathbb{Z}}$ are iid, we obtain that

\begin{align*}
  &\E \| G_{2,n}(u) - \E G_{2,n}(u)\|^2 \\
  &= \E \left\| \frac{1}{n} \sum_{i=1}^{n} \left[\frac{X_{i+h}}{\|X_i - n^{-1/2}u\|\|X_{i+h} - n^{-1/2}u\|} - \E G_{2,n}(u) \right] \right\|^2   \\
  &=\frac{1}{n^2} \Biggl[  \sum_{i=1}^{n}  \E \left\|  \frac{X_{i+h}}{\|X_i - n^{-1/2}u\|\|X_{i+h} - n^{-1/2}u\|} - \E G_{2,n}(u)  \right\|^2 \\
  &+ \sum_{1\le i \ne j \le n}  \E \left\langle   \frac{X_{i+h}}{\|X_i - n^{-1/2}u\|\|X_{i+h} - n^{-1/2}u\|} - \E G_{2,n}(u) ,\right.\\
  &\left. \hspace{+2cm} \frac{X_{j+h}}{\|X_j - n^{-1/2}u\|\|X_{j+h} - n^{-1/2}u\|} - \E G_{2,n}(u) \right \rangle \Biggl].
\end{align*}
When $i\ne j$, and $|j-i|\ne h$, it follows due to the assumption that the $\{X_i\}_{i\in \mathbb{Z}}$ are iid that
\begin{align*}
&\E \left\langle   \frac{X_{i+h}}{\|X_i - n^{-1/2}u\|\|X_{i+h} - n^{-1/2}u\|} - \E G_{2,n}(u) , \right.\\
&\left. \hspace{0.5cm}  \frac{X_{j+h}}{\|X_j - n^{-1/2}u\|\|X_{j+h} - n^{-1/2}u\|} - \E G_{2,n}(u) \right \rangle=0.
\end{align*}
And when $|j-i|=h$, by the Cauchy-Schwarz inequality the same inner product term is bounded by
\begin{align}\label{normterm}
\E \left\|  \frac{X_{ h}}{\|X_0 - n^{-1/2}u\|\|X_{h} - n^{-1/2}u\|} - \E G_{2,n}(u)  \right\|^2.
\end{align}
Hence

\begin{align}\label{variance}
&\E \| G_{2,n}(u) - \E G_{2,n}(u)\|^2 \nonumber \\
&\le \frac{C_7}{n} \E \left\|  \frac{X_{ h}}{\|X_0 - n^{-1/2}u\|\|X_{h} - n^{-1/2}u\|} - \E G_{2,n}(u)  \right\|^2.
\end{align}

By Assumption \ref{as-moment-2} and Lemma \ref{lem-sqnorm}, for all $n$ sufficiently large,
\begin{align*}
&\E \left\|  \frac{X_{ h}}{\|X_0 - n^{-1/2}u\|\|X_{h} - n^{-1/2}u\|}  \right\|^2 \\
&= \E \|X_0 - n^{-1/2}u\|^{-2} \E \left\|  \frac{X_{ h}}{\|X_{h} - n^{-1/2}u\|}\right\|^2  \le C_8,
\end{align*}
which implies the term in \eqref{normterm} is bounded for all $n$ sufficiently large. This along with \eqref{variance} gives that

\begin{align}\label{s-1}
  G_{2,n}(u) \stackrel{P}{\to} \E G_{2,n}(u), \mbox{ for each } u \in K.
\end{align}
We now aim to show $G_{2,n}$ is stochastically Lipshitz. To that end, for $u,u' \in K$,  simple arithmetic gives that

\begin{align*}
&\|G_{2,n}(u)- G_{2,n}(u')\| \\
&=   \Biggl\| \frac{1}{n} \sum_{i=1}^{n} \left[\frac{1}{\|X_i - n^{-1/2}u\|} - \frac{1}{\|X_i - n^{-1/2}u'\|} \right]\frac{X_{i+h}}{\|X_{i+h} - n^{-1/2}u\|}     \\
&\;\;\;\; +\frac{1}{n} \sum_{i=1}^{n} \left[\frac{X_{i+h}}{\|X_i - n^{-1/2}u\|} -\frac{X_{i+h}}{\|X_i - n^{-1/2}u'\|}  \right] \\
& \hspace{+2cm} \times \frac{1}{\|X_i - n^{-1/2}u'\|}\frac{X_{i+h}}{\|X_i - n^{-1/2}u\|} \Biggl\| \\
&\le    \Biggl\| \frac{1}{n} \sum_{i=1}^{n} \left[\frac{1}{\|X_i - n^{-1/2}u\|} - \frac{1}{\|X_i - n^{-1/2}u'\|} \right]\frac{X_{i+h}}{\|X_i - n^{-1/2}u\|} \Biggl\|    \\
&\;\;\;\; + \Biggl\|\frac{1}{n} \sum_{i=1}^{n} \left[\frac{X_{i+h}}{\|X_i - n^{-1/2}u\|} -\frac{X_{i+h}}{\|X_i - n^{-1/2}u'\|}  \right] \\
& \hspace{+2cm} \times \frac{1}{\|X_i - n^{-1/2}u'\|}\frac{X_{i+h}}{\|X_i - n^{-1/2}u\|} \Biggl\| \\
&=: \| L_{1,n}(u,u')\| + \|L_{2,n}(u,u')\|.
\end{align*}
Notice that by the triangle inequality,

\begin{align*}
 &\| L_{1,n}(u,u')\| \\
 &\le \frac{1}{n} \sum_{i=1}^{n}  \frac{n^{-1/2}\|u-u'\|}{\|X_i - n^{-1/2}u\|\|X_i - n^{-1/2}u'\|} \left\| \frac{X_{i+h}}{\|X_{i+h} - n^{-1/2}u\|} \right\| \\
 &=\|u-u'\| \frac{1}{n^{3/2}} \sum_{i=1}^{n}  \frac{1}{\|X_i - n^{-1/2}u\|\|X_i - n^{-1/2}u'\|} \left\| \frac{X_{i+h}}{\|X_{i+h} - n^{-1/2}u\|} \right\| \\
 &=:  \|u-u'\|B_n,
\end{align*}
where by Assumption \ref{as-moment-2} and Lemma \ref{lem-sqnorm}, $\E B_n \le C_9$. This implies that the process $\{G_{2,n}(u),\; u \in K\}$ satisfies the equicontinuity conditions of Assumption 3A of Newey (1991). This along with \eqref{s-1} imply according to Theorem 2.1 of Newey (1991) that
$$
\sup_{u \in K} \|G_{2,n}-\E G_{2,n}\| = o_P(1).
$$
In conjunction with \eqref{s-2}, this shows that $\sup_{u \in K} \|G_{2,n}(u)\|=o_P(1)$. Due to \eqref{s4}, \eqref{s5}, and \eqref{s-7}, it follows that  $R_{2,n}=o_P(1)$. Following a nearly identical argument, it may be shown that $R_{3,n}=o_P(1)$ and $R_{4,n}=o_P(1)$. Although the argument is also similar to show that $R_{1,n}=o_P(1)$, it differs in a few places and so we think it is valuable to provide the details. Notice that we may write

\begin{align*}
R_{1,n} &= \frac{1}{\sqrt{n}} \sum_{i=1}^{n} \frac{1}{\hat{m}_{i,h}} \left[ \langle  X_i  \|X_i\|, X_{i+h}  \|X_{i+h}\| \rangle \right.\\
&\;\;\;\;\;\;\left. - \langle X_i\|X_i - \hat{\mu}\|, X_{i+h}\|X_{i+h}- \hat{\mu}\| \rangle   \right] \\
 &= \frac{1}{\sqrt{n}} \sum_{i=1}^{n} \frac{1}{\hat{m}_{i,h}} \langle  X_i  \|X_i\|, X_{i+h} [ \|X_{i+h}\|-\|X_{i+h}- \hat{\mu}\|]   \rangle \\
 &\;\;\;\;\;\;+ \frac{1}{\sqrt{n}} \sum_{i=1}^{n} \frac{1}{\hat{m}_{i,h}}  \langle  X_i  [\|X_i\|-\|X_{i }- \hat{\mu}\|], X_{i+h} \|X_{i+h}- \hat{\mu}\|    \rangle \\
  &= \frac{1}{\sqrt{n}} \sum_{i=1}^{n} \frac{ \langle  X_i  \|X_i\|, X_{i+h} [ \|X_{i+h}\|-\|X_{i+h}- \hat{\mu}\|]   \rangle}{\|X_{i+h}\| \|X_{i}- \hat{\mu}\|\| X_{i+h}- \hat{\mu}\|} \\
 &\;\;\;\;\;\;+ \frac{1}{\sqrt{n}} \sum_{i=1}^{n} \frac{ \langle  X_i  [\|X_i\|-\|X_{i }- \hat{\mu}\|], X_{i+h} \|X_{i+h}- \hat{\mu}\|    \rangle}{{\|X_{i+h}\| \|X_{i}- \hat{\mu}\|\| X_{i}\|}} \\
&=: R_{1,n}^{(1)} + R_{1,n}^{(2)}.
\end{align*}

Define

$$
G_{1,n}^{(1)}(u) =  \frac{1}{\sqrt{n}} \sum_{i=1}^{n} \frac{ \langle  X_i  \|X_i\|, X_{i+h} [ \|X_{i+h}\|-\|X_{i+h}- n^{-1/2}u\|]   \rangle}{\|X_{i+h}\| \|X_{i}- n^{-1/2}u\|\| X_{i+h}- n^{-1/2}u\|}
$$
so that $R_{1,n}^{(1)}= G_{1,n}^{(1)}(\sqrt{n}\hat{\mu})$. Arguing as in \eqref{s-7}, it is enough to show that for any compact $K \subset \mathcal{H}$,

$$
\sup_{u \in K} |G_{1,n}^{(1)}(u)| =o_P(1).
$$

Due the independence of the $X_i's$ and Fubini's Theorem,

\begin{align*}
  \E G_{1,n}^{(1)}(u)   = \sqrt{n} \left\langle \E \frac{X_0}{\|X_0 - n^{-1/2}u\| }, \E \frac{X_h[ \| X_h \| - \|X_h - n^{-1/2}u]|}{\|X_h - n^{-1/2}u\|\|X_h\| } \right\rangle.
\end{align*}

By the Cauchy-Schwarz and reverse triangle inequality, \eqref{mean-conv-1}, and Assumption \ref{as-moment-2}, it follows that

\begin{align}\label{s-8}
  \E G_{1,n}^{(1)}(u)   &\le  \sqrt{n} \left\| \E \frac{X_0}{\|X_0 - n^{-1/2}u\| }\right\| \left\| \E \frac{X_h[ \| X_h \| - \|X_h - n^{-1/2}u]|}{\|X_h - n^{-1/2}u\|\|X_h\| } \right\| \\
  &\le C_{10}  \left\| \E \frac{X_0}{\|X_0 - n^{-1/2}u\| }\right\| \E \|X_h - n^{-1/2}u\|^{-1} \to 0, \notag
\end{align}
as $n\to \infty$ for each $u\in K$. Moreover, noting that $(X_i,X_{i+h})$ is independent of $(X_j,X_{j+h})$ when $i\ne j$, $|i-j|\ne h$, we obtain along with the Cauchy-Schwarz inequality that

\begin{align*}
  \mbox{Var}(G_{1,n}^{(1)}(u)) &= \frac{1}{n}\Biggl[ \sum_{i=1}^{n} \mbox{Var}\left( \frac{ \langle  X_i  \|X_i\|, X_{i+h} [ \|X_{i+h}\|-\|X_{i+h}- \hat{\mu}\|]   \rangle}{\|X_{i+h}\| \|X_{i}- \hat{\mu}\|\| X_{i+h}- \hat{\mu}\|} \right) \\
  &+ 2 \sum_{1 \le i<j \le n} \mbox{Cov}\left( \frac{ \langle  X_i  \|X_i\|, X_{i+h} [ \|X_{i+h}\|-\|X_{i+h}- \hat{\mu}\|]   \rangle}{\|X_{i+h}\| \|X_{i}- \hat{\mu}\|\| X_{i+h}- \hat{\mu}\|}, \right.\\
  &\left. \hspace{+2cm} \frac{ \langle  X_j  \|X_j\|, X_{j+h} [ \|X_{j+h}\|-\|X_{j+h}- \hat{\mu}\|]   \rangle}{\|X_{j+h}\| \|X_{j}- \hat{\mu}\|\| X_{j+h}- \hat{\mu}\|} \right) \Biggl] \\
  &\le C_{11} \mbox{Var}\left( \frac{ \langle  X_i  \|X_i\|, X_{i+h} [ \|X_{i+h}\|-\|X_{i+h}- \hat{\mu}\|]   \rangle}{\|X_{i+h}\| \|X_{i}- \hat{\mu}\|\| X_{i+h}- \hat{\mu}\|} \right).
\end{align*}
Again applying Assumption \ref{as-moment-2} and the reverse triangle inequality, we obtain that

$$
\mbox{Var}\left( \frac{ \langle  X_i  \|X_i\|, X_{i+h} [ \|X_{i+h}\|-\|X_{i+h}- \hat{\mu}\|]   \rangle}{\|X_{i+h}\| \|X_{i}- \hat{\mu}\|\| X_{i+h}- \hat{\mu}\|} \right) \le C_{12} n^{-1/2} \to 0 \mbox{ as $n\to \infty$}.
$$
This implies by Chebyshev's inequality that $G_{1,n}^{(1)}(u)- \E G_{1,n}^{(1)}(u)=o_P(1).$ In order to establish the stochastic equicontinuity of $G_{1,n}^{(1)}$, we note that for $u,u' \in K$,

\begin{align*}
&| G_{1,n}^{(1)}(u) - G_{1,n}^{(1)}(u')| = \Biggl| \frac{1}{\sqrt{n}} \sum_{i=1}^{n} \frac{ \langle  X_i  \|X_i\|, X_{i+h} [ \|X_{i+h}\|-\|X_{i+h}- n^{-1/2}u\|]   \rangle}{\|X_{i+h}\| \|X_{i}- n^{-1/2}u\|\| X_{i+h}- n^{-1/2}u\|}  \\
&\;\;\;\;\;\;\;\;\;\;\;\;\;\;\;\;\;\;\;\;\;\;\;\;\;\;\;-  \frac{ \langle  X_i  \|X_i\|, X_{i+h} [ \|X_{i+h}\|-\|X_{i+h}- n^{-1/2}u'\|]   \rangle}{\|X_{i+h}\| \|X_{i}- n^{-1/2}u'\|\| X_{i+h}- n^{-1/2}u'\|} \Biggl| \\
&\le \frac{1}{\sqrt{n}} \sum_{i=1}^{n} \Biggl[ \left| \frac{ 1 }{ \|X_{i}- n^{-1/2}u\|\| X_{i+h}- n^{-1/2}u\|} \right.\\
&\hspace{3cm} \left. - \frac{ 1 }{ \|X_{i}- n^{-1/2}u'\|\| X_{i+h}- n^{-1/2}u'\|} \right| \\
& \hspace{2cm} \times \Biggl|\Biggl\langle  X_i  \|X_i\|, X_{i+h} \frac{[\|X_{i+h}- n^{-1/2}u\|-\|X_{i+h}- n^{-1/2}u'\|]}{\|X_{i+h}\|}    \Biggl\rangle \Biggl| \Biggl] \\
&\;\;\;\;+ \frac{1}{\sqrt{n}} \sum_{i=1}^{n} \Biggl[ \left| \frac{ 1 }{ \|X_{i}- n^{-1/2}u\|\| X_{i+h}- n^{-1/2}u\|} \right.\\
&\hspace{3cm} \left. - \frac{ 1 }{ \|X_{i}- n^{-1/2}u'\|\| X_{i+h}- n^{-1/2}u'\|} \right| \\
& \hspace{2cm} \times \Biggl|\Biggl\langle  X_i  \|X_i\|, X_{i+h} \frac{[ \|X_{i+h}\|-\|X_{i+h}- n^{-1/2}u\|]}{\|X_{i+h}\|}    \Biggl\rangle \Biggl| \Biggl] \\
&=: L_{1,n}(u,u') + L_{2,n}(u,u').
\end{align*}
By Cauchy-Schwarz and the reverse triangle inequalities,

$$
L_{1,n}(u,u') \le \|u-u'\| \frac{1}{n} \sum_{i=1}^{n} \frac{\|X_i\|}{\|X_i - n^{-1/2}u\|\|X_{i+h} - n^{-1/2}u'\|}=:\|u-u'\| B_{n,1}.
$$
By Assumption \ref{as-moment-2} and Lemma \ref{lem-sqnorm}, $EB_{1,n}<\infty$. It follows similarly that

$$
L_{2,n}(u,u') \le \|u-u'\| B_{n,2}, \mbox{ with } \E B_{2,n}<\infty.
$$
Hence $| G_{1,n}^{(1)}(u) - G_{1,n}^{(1)}(u')| \le \|u-u'\| B_n^*,$ where $EB_n^*< \infty$. It follows once again by Theorem 2.1 of Newey (1991) that
$$
\sup_{u\in K} |G_{1,n}^{(1)}(u) - \E G_{1,n}^{(1)}(u)|=o_P(1).
$$
Owing then to \eqref{s-8}, $\sup_{u \in K} |G_{1,n}^{(1)}(u)|=o_P(1).$ It follows similarly that $\sup_{u \in K} |G_{1,n}^{(2)}(u)|=o_P(1),$ and hence $R_{1,n}=o_P(1),$ giving the result.
\end{proof}

\subsection{Proof of Theorem \ref{thm-cons-2}}\label{app-thm2}

\begin{proof}
We first note that
\begin{align*}
	\abs{\frac{1}{n}\sum_{i=1}^{n-h} \inp{S(X_i-\hmu)}{S(X_{i+h}-\hmu)} - \frac{1}{n}\sum_{i=1}^{n} \inp{S(X_i-\hmu)}{S(X_{i+h}-\hmu)}}  &\le \frac{h}{n} \\
	&\ias 0 .
\end{align*}

	Thus the theorem follows upon showing that
	\[
	\abs{\frac{1}{n}\sum_{i=1}^n \inp{S(X_i-\hmu)}{S(X_{i+h}-\hmu)} - \frac{1}{n}\sum_{i=1}^n \inp{S(X_i)}{S(X_{i+h})}} \ias 0.
	\]
	
	Let $A_n = \{(x,y) \in \HH \times \HH: 2\norm{x-\hmu}> \norm{x} \text{ and }  2\norm{y-\hmu}> \norm{y}\}$. Then 
	
	\begin{align*}
		&\abs{\frac{1}{n}\sum_{i=1}^n \inp{S(X_i-\hmu)}{S(X_{i+h}-\hmu)} - \frac{1}{n}\sum_{i=1}^n \inp{S(X_i)}{S(X_{i+h})}} \\
		&=
		\frac{1}{n}\sum_{i=1}^n \abs{\inp{S(X_i-\hmu)}{S(X_{i+h}-\hmu)}  - \inp{S(X_i)}{S(X_{i+h})}} \\
		&=
		\frac{1}{n}\san \abs{\inp{S(X_i-\hmu)}{S(X_{i+h}-\hmu)}  - \inp{S(X_i)}{S(X_{i+h})}} \\
		&+
		\frac{1}{n}\sanc \abs{\inp{S(X_i-\hmu)}{S(X_{i+h}-\hmu)}  - \inp{S(X_i)}{S(X_{i+h})}} \\
		&=:
		S_{n,1} + S_{n,2}.
	\end{align*}
	We aim to show $S_{n,1}\ias0$ and $S_{n,2}\ias0$. Looking at the first term $S_{n,1}$, we have
	\begin{align*}
		&S_{n,1}\\
		&= 
		\frac{1}{n}\san \abs{\inp{\frac{\xihmu}{\nxihmu}}{\frac{\xihhmu}{\nxihhmu}} - \inp{\frac{X_i}{\nxi}}{\frac{\xih}{\nxih}}} \\
		&=
		\frac{1}{n}\san \Biggl| \inp{\frac{(\xihmu)\nxi}{\nxihmu\nxi}}{\frac{(\xihhmu)\nxih}{\nxihhmu\nxih}} \\
		 & \;\;\;\;\;\;\;\;\;\; - \inp{\frac{X_i\nxihmu}{\nxi\nxihmu}}{\frac{\xih\nxihhmu}{\nxih\nxihhmu}} \Biggl| \\
		&= \frac{1}{n}\san\frac{1}{\hat{m}_{i,h}}\times \\
		&\abs{\inp{(\xihmu)\nxi}{(\xihhmu)\nxih} - \inp{X_i(\nxihmu)}{\xih(\nxihhmu)}},
	\end{align*}
	where again $\hat{m}_{i,h}= \nxi \nxih \nxihmu \nxihhmu$. By expanding the summand and applying the triangle inequality, we get that 
	\begin{align*}
		&
		\frac{1}{\hat{m}_{i,h}}\abs{\inp{(\xihmu)\nxi}{(\xihhmu)\nxih} \right.\\
		&\;\;\;\left. - \inp{X_i(\nxihmu)}{\xih(\nxihhmu)}}  \\
		&= \frac{1}{\hat{m}_{i,h}}\abs{\left\{\inp{X_i\nxi}{\xih\nxih} - \inp{X_i(\nxihmu)}{\xih(\nxihhmu)}\right\} \right.\\
		&\;\;\;\left. - \inp{\hmu\nxi}{\xih\nxih} 
			- \inp{X_i\nxi}{\hmu\nxih} + \inp{\hmu\nxi}{\hmu\nxih}} \\
		&=
		\frac{1}{\hat{m}_{i,h}}\abs{S_{n,i}^{(1)}-S_{n,i}^{(2)}-S_{n,i}^{(3)}+S_{n,i}^{(4)}} \le  \frac{1}{\hat{m}_{i,h}}\left(\abs{S_{n,i}^{(1)}} + \abs{S_{n,i}^{(2)}} + \abs{S_{n,i}^{(3)}} +\abs{S_{n,i}^{(4)}}\right),
	\end{align*}
	where 
	\begin{align*}
		S_{n,i}^{(1)} &= \inp{X_i\nxi}{\xih\nxih} - \inp{X_i(\nxihmu)}{\xih(\nxihhmu)},\\
		S_{n,i}^{(2)} &= \inp{\hmu\nxi}{\xih\nxih},  \\	
		S_{n,i}^{(3)}  &= \inp{X_i\nxi}{\hmu\nxih}, \;\; \mbox{and } S_{n,i}^{(4)}  =  \inp{\hmu\nxi}{\hmu\nxih}.
	\end{align*}
	
	By adding and subtracting $\inp{X_i\nxi}{\nxih\nxihhmu}$ and using the triangle inequality, we have 
	\begin{align*}
		&\frac{1}{\hat{m}_{i,h}}\abs{S_{n,i}^{(1)}} \\
		&= \frac{1}{\hat{m}_{i,h}}\abs{ \left\{\inp{X_i\nxi}{\xih\nxih} - \inp{X_i\nxi}{\xih\nxihhmu}\right\} \right.\\
	    &\left.
			+ \left\{\inp{X_i\nxi}{\xih\nxihhmu} - \inp{X_i(\nxihmu)}{\xih(\nxihhmu)}\right\}} \\
		&=:
		\frac{1}{\hat{m}_{i,h}}\abs{S_{n,i}^{(1,1)} + S_{n,i}^{(1,2)}} \le \frac{1}{\hat{m}_{i,h}} \abs{S_{n,i}^{(1,1)}} + \frac{1}{\hat{m}_{i,h}} \abs{S_{i,i}^{(1,2)}},
	\end{align*}
	where 
	\begin{align*}
		S_{n,i}^{(1,1)} &= \inp{X_i\nxi}{\xih\nxih} - \inp{X_i\nxi}{\xih\nxihhmu} \\
		S_{n,i}^{(1,2)} &= \inp{X_i\nxi}{\xih\nxihhmu} - \inp{X_i(\nxihmu)}{\xih(\nxihhmu)}.
	\end{align*}
	
	Using the reverse triangle inequality and Cauchy–Schwarz inequality, we obtain that 
	\begin{align*}
		\frac{1}{\hat{m}_{i,h}}\abs{S_{n,1}^{(1,1)}}
		&= 
		\frac{1}{\hat{m}_{i,h}}\abs{\inp{X_i\nxi}{\xih(\nxih-\nxihhmu)}} \\
		&= 
		\frac{(\nxih-\nxihhmu)\nxi}{\hat{m}_{i,h}}\inp{X_i}{\xih}\\
		& \le  \frac{\nhmu\nxi}{\hat{m}_{i,h}}\inp{X_i}{\xih}  \le \frac{\nhmu\nxi^2\nxih}{\hat{m}_{i,h}} \\
		&\le \frac{\nhmu\nxi^2\nxih}{\nxi\nxih\nxihmu\nxihhmu} \\
		&=
		\frac{\nxi\nhmu}{\nxihmu\nxihhmu}. 
	\end{align*}
	For $(X_i,\xih)\in A_n$, it follows that 
	\[
	\frac{\nxi\nhmu}{\nxihmu\nxihhmu}  \le 4\frac{\nxi\nhmu}{\nxi\nxih}= \frac{4\nhmu}{\nxih}.
	\]
	Combining the above arguments, it follows that 
	\begin{align*}
		&\frac{1}{n}\san \frac{1}{\hat{m}_{i,h}}\abs{S_{n,i}^{(1,1)}} \le \frac{1}{n} \san\frac{4}{\nxih}\nhmu\\
		&= 4\nhmu \frac{1}{n} \san\frac{1}{\nxih}\le 4\nhmu \frac{1}{n}\sum_{i=1}^n \frac{1}{\nxih}.
	\end{align*}
	Since $\nhmu \ias 0$ and $\frac{1}{n} \sum_{i=1}^n\nxih^{-1}\ias E{\nxih^{-1}}<\infty$ by the ergodic Theorem, we have that 
	\begin{equation}\label{eq-s111}
		\frac{1}{n} \san\frac{1}{\hat{m}_{i,h}} \abs{S_{n,i}^{(1,1)}} \ias 0.
	\end{equation}
	A similar argument shows that 
	\begin{equation}\label{eq-s112}
		\frac{1}{n} \san \frac{1}{\hat{m}_{i,h}}\abs{S_{n,i}^{(1,2)}} \ias 0.
	\end{equation}
	Combining \eqref{eq-s111} and \eqref{eq-s112}, we get that 
	\[
	\frac{1}{n} \san \frac{1}{\hat{m}_{i,h}}\abs{S_{n,i}^{(1)}} \ias 0.
	\] 
	
	As for $S_{n,i}^{(2)}$, using the Cauchy-Schwarz inequality, we have that 
	\begin{align*}
		\frac{1}{\hat{m}_{i,h}}\abs{S_{n,i}^{(2)}}
		&= 
		\frac{1}{\hat{m}_{i,h}}\abs{\inp{\hmu\nxi}{\xih\nxih}} \\
		&\le \frac{\nhmu\nxi\nxih^2}{\nxi\nxih\nxihmu\nxihhmu}  \\
		&=
		\frac{\nhmu\nxih}{\nxihmu\nxihhmu}.
	\end{align*}
	For $(X_i,X_{i+h})\in A_n$, it follows that 
	\[
	\frac{\nhmu\nxih}{\nxihmu\nxihhmu} \le \frac{4\nhmu\nxih}{\nxi\nxih} = 4 \frac{\nhmu}{\nxi}.
	\]
	
	Thus,
	\[
	\frac{1}{n}\san \frac{1}{\hat{m}_{i,h}}\abs{S_{n,i}^{(2)}} \le  \frac{1}{n}\sum_{i=1}^n \frac{1}{\hat{m}_{i,h}}\abs{S_{n,i}^{(2)}} \le 
	4\nhmu\frac{1}{n}\sum_{i=1}^n\frac{1}{\nxi} \ias 0.
	\]
	
	Nearly identical arguments give that 
	\[
	\frac{1}{n}\san\frac{1}{\hat{m}_{i,h}}\abs{S_{n,i}^{(3)}}\ias 0.
	\]
	Finally considering $S_{n,i}^{(4)}$, we have that using the Cauchy-Schwarz inequality that 
	\begin{align*}
		&\frac{1}{n}\san \abs{S_{n,i}^{(4)}} = \frac{1}{n}\san \frac{1}{\hat{m}_{i,h}}\abs{\inp{\hmu\nxi}{\hmu\nxih}}\\
		&\le 
		\frac{1}{n}\san \frac{\nhmu^2 \nxi\nxih}{\nxi\nxih\nxihmu\nxihhmu}\\
		&\le 
		4\nhmu^2 \frac{1}{n} \san \frac{1}{\nxi\nxih} \\
		&\le 
		4\nhmu^2 \frac{1}{n}\sum_{i=1}^n \frac{1}{\nxi\nxih} \ias 0,
	\end{align*}
	since $\nhmu^2 \ias 0$, and $\frac{1}{n}\sum_{i=1}^n \nxi^{-1}\nxih^{-1} \ias E[\nxi^{-1}\nxih^{-1}] <\infty$ by Assumption \ref{ass2}. Thus, we conclude that
	\[
	S_{n,i}^{(1)}=\frac{1}{n} \san \abs{\inp{S(X_i-\hmu)}{S(X_{i+h}-\hmu)}  - \inp{S(X_i)}{S(X_{i+h})}} \ias 0.
	\]
	
	Turning to $S_{n,2}$. The argument here is similar to that presented in the proof of Theorem 1 of \cite{boente:2019:spatial}, but we provide the details so that the proof is self-contained. First note that 
	\begin{equation}\label{eq-norm}
		\norm{\inp{S(\xihmu)}{S(\xihhmu)}} = \norm{\inp{S(X_i)}{S(\xihhmu)}} = 1.
	\end{equation}
	
	Let $d_{n,x} = \{(x,y)\in \HH\times \HH:2 \norm{x-\hmu} \ge \norm{x}\}$, $d_{n,y} = \{(x,y)\in \HH \times \HH:2 \norm{y-\hmu} \ge \norm{y}\}$.

	Recall $A_n = \{(x,y) \in \HH \times \HH: 2\norm{x-\hmu}> \norm{x} \text{ and }  2\norm{y-\hmu}> \norm{y}\}$. Consequently
	$A_n^c = \{(x,y) \in \HH \times \HH : 2\norm{x-\hmu}\le \norm{x}~ \text{or}~ 2\norm{y-\hmu}\le \norm{y} \}$. Clearly $A_n^c =  d_{n,x}^c \cup d_{n,y}^c$.
	Consequently,
	\begin{equation}\label{eq-z}
		\1{A_n^c}(x,y)  \le \1{d_{n,x}^c}(x,y) + \1{d_{n,y}^c}(x,y).
	\end{equation}
	
	Using the triangle inequality that 
	\begin{align}\label{eq-sn2}
		S_{n,2} &= \frac{1}{n} \sanc\norm{\inp{S(X_i-\hmu)}{S(X_{i+h}-\hmu)} -  \inp{S(X_i)}{S(X_{i+h})}} \nonumber\\
		&\le \frac{2}{n} \sanc 1= \frac{2}{n} \sum_{i=1}^n \1{A_n^c}(X_i,X_{i+h}) \nonumber \\
		&\le \frac{2}{n} \sum_{i=1}^n\left(\1{d_{n,x}^c}(X_i,X_{i+h}) + \1{d_{n,y}^c}(X_i,X_{i+h}) \right) \nonumber \\
		&=  \frac{2}{n} \sum_{i=1}^n \1{d_{n,x}^c}(X_i,X_{i+h}) +  \frac{2}{n} \sum_{i=1}^n \1{d_{n,y}^c}(X_i,X_{i+h}) \nonumber\\
		&=: J_{n,1} + J_{n,2}.
	\end{align}
	We aim to show that $J_{n,1}\ias 0 $ and $J_{n,2}\ias0$.
	
	Note that the assumption $\EXp{\norm{X}^{-1}}<\infty$ implies that $P(\norm{X} = 0 ) = 0$. Thus, by continuity of the probability measure, for every $\epsilon>0$, there exists $\delta>0$, such that $P(\norm{X}\in B_{\delta,x})\le \epsilon$, where $B_{\delta,x} = \{(X,Y)\in \HH \times \HH : \norm{X}\le \delta\}$. Then
	\begin{align*}
		&\frac{1}{n}\sum_{i=1}^n \1{d_{n,x}^c}(X_i, X_{i+h}) \\
		& \le \frac{1}{n} \sum_{i=1}^n \1{B_{\delta,x}}(X_i,X_{i+h}) + \frac{1}{n} \sum_{i=1}^n \left(  \1{d_{n,x}^c}(X_i, X_{i+h}) - \1{B_{\delta,x}}(X_i,X_{i+h}) \right)_+ \\
		&=: J_{n,1}^{(1)} + J_{n,1}^{(2)},
	\end{align*}
	where $a_+ = \max(a,0)$. Using the ergodic theorem, $J_{n,1}^{(1)} \ias P(\norm{X}\in B_{\delta,x})\le \epsilon$. Since $\epsilon$ is arbitrary, we take $\epsilon \downarrow 0^+$ and obtained $J_{n,1}^{(1)}\ias 0$. As for the $J_{n,2}$, note that $\{(X_i,X_{i+h})\in \HH \times \HH : \nhmu \le \delta/2\}\subset \{(X_i,X_{i+h})\in \HH \times \HH : \1{d_{n,x}^c}(X_i, X_{i+h}) - \1{B_{\delta,x}}(X_i,X_{i+h})=0 \}$. Hence, using Assumption \ref{ass1}, there exists a null set $N$ such that for $\omega\notin N$, there exists an $M\in \mathbb{Z}$ such that for every integer $n\ge M$, $\nhmu\le \delta/2$. This implies
	\[
	J_{n,1}^{(2)} = \frac{2}{n} \sum_{i=1}^n \1{d_{n,y}^c}(X_i,X_{i+h}) = 0.
	\] 
	Combining the arguments above, we obtain $ 	J_{n,1} \ias 0.$ An identical argument shows also that $J_{n,2} \ias 0.$ Therefore, $S_{n,2} \ias 0,$ which completes the proof.
\end{proof}

\section{Additional Simulation Results}\label{sim-app}

\subsection{Accuracy of estimator \texorpdfstring{$\|\hat{C}_P\|_2$}{norm of C\_p}}
We examined the accuracy of variance estimator $\|\hat{C}_P\|_2$ by generating data for which the value $\|{C}_P\|_2$ may be computed explicitly. Consider the two-dimensional process 
	\[
	X_i(t) = \sqrt{2\lambda_1}Z_{1,i} \sin(2\pi t) + \sqrt{2\lambda_2} Z_{2,i} \cos(2\pi t),
	\]
	where $Z_{j,i}$ are a family of independent and identically distributed $\N(0,1)$ random variables. Straightforward calculation shows that $\|C_P\|_2$ may be calculated based on the moments ratios of quadratic forms of normal random variables. Such moments are computed explicitly in \cite{magnus:1986:quadratic}. It may be shown when $\lambda_1=\lambda_2=1$, then 
	
	 covariance kernel of the projected data is
	\begin{align*}
		C_P(t,s) = \cos(2\pi t) \cos(2\pi s) + \sin(2\pi t) \sin(2\pi s)=\cos(2\pi(s-t)),
	\end{align*}
	and hence in this case $\|C_P\|_2^2=1/2$. Similarly, when $\lambda_1 = 1,\lambda_2=2$,  $\|C_P\|_2^2 = 9 - 6\sqrt{2}.$
	
	In order to examine the effect of variance estimation the coverage of the confidence intervals in \eqref{conf-def}, we considered four different ways for constructing the confidence interval: 1) Normal confidence interval with true variance; 2) t-confidence interval with true variance with degrees of freedom $n-1$; 3) normal confidence interval with estimated variance and 4) t-confidence interval with estimated variance with degrees of freedom $n-1$. The empirical coverage rates for each of these intervals are reported in Table \ref{tab-finite}. We observed that the difference between the four construction methods yields was quite small, with the t-confidence interval being relatively narrow which is a phenomenon that is widely known. With the coverage rates that are close to the nominal levels, it evidently supports the derived asymptotic distribution given in \eqref{eq-distribution}.
	
	\begin{table}[hbp]
		\centering
		\caption{The comparison between the variance estimate in a finite-dimensional case: 1) $Z_{\alpha/2} \|C_P\|_2^2$ is normal confidence interval with true variance; 2) $t_{\alpha/2,n-1} \|C_P\|_2^2$ is t-confidence interval with true variance with degrees of freedom $n-1$; 3) $Z_{\alpha/2} \|\hat{C}_P\|_2^2$ is normal confidence interval with estimated variance; and 4) $t_{\alpha/2,n-1} \|\hat{C}_P\|_2^2$ is t-confidence interval with estimated variance with degrees of freedom $n-1$.}
		\label{tab-finite}
		\begin{adjustbox}{max width=\textwidth}
			\begin{tabular}{@{}lllllllllllllll@{}}
				\toprule
				\multicolumn{7}{c}{$\lambda_1 = 1$ and $\lambda_2 = 2$}                 &  & \multicolumn{7}{c}{$\lambda_1 = \lambda_2=1$}                   \\ \midrule
				N   & $\alpha$ & h  &  $Z_{\alpha/2} \|C_P\|_2^2$ & $t_{\alpha/2, n-1}\|C_P\|_2^2$ & $Z_{\alpha/2} \|\hat{C}_P\|_2^2$& $t_{\alpha/2, n-1} \|\hat{C}_P\|_2^2$ &  & N   & $\alpha$ & h  & $Z_{\alpha/2} \|C_P\|_2^2$ & $t_{\alpha/2, n-1} \|C_P\|_2^2$ & $Z_{\alpha/2} \|\hat{C}_P\|_2^2$ & $t_{\alpha/2, n-1} \|\hat{C}_P\|_2^2$ \\
				\midrule
				100 & 0.01  & 1  & 0.008  & 0.006  & 0.007 & 0.006 &  & 100 & 0.01  & 1  & 0.007  & 0.006  & 0.007 & 0.006 \\
				&       & 5  & 0.008  & 0.006  & 0.007 & 0.006 &  &     &       & 5  & 0.009  & 0.007  & 0.009 & 0.006 \\
				&       & 10 & 0.009  & 0.008  & 0.008 & 0.008 &  &     &       & 10 & 0.008  & 0.007  & 0.007 & 0.006 \\
				& 0.05  & 1  & 0.048  & 0.046  & 0.045 & 0.045 &  &     & 0.05  & 1  & 0.047  & 0.046  & 0.047 & 0.043 \\
				&       & 5  & 0.049  & 0.046  & 0.047 & 0.044 &  &     &       & 5  & 0.048  & 0.045  & 0.047 & 0.045 \\
				&       & 10 & 0.044  & 0.041  & 0.041 & 0.038 &  &     &       & 10 & 0.045  & 0.04   & 0.043 & 0.04  \\
				& 0.10  & 1  & 0.09   & 0.087  & 0.09  & 0.088 &  &     & 0.10  & 1  & 0.094  & 0.093  & 0.094 & 0.09  \\
				&       & 5  & 0.102  & 0.099  & 0.103 & 0.098 &  &     &       & 5  & 0.11   & 0.105  & 0.109 & 0.105 \\
				&       & 10 & 0.088  & 0.084  & 0.086 & 0.084 &  &     &       & 10 & 0.094  & 0.091  & 0.092 & 0.088 \\
				250 & 0.01  & 1  & 0.011  & 0.011  & 0.011 & 0.011 &  & 250 & 0.01  & 1  & 0.008  & 0.008  & 0.008 & 0.008 \\
				&       & 5  & 0.008  & 0.007  & 0.007 & 0.006 &  &     &       & 5  & 0.006  & 0.006  & 0.006 & 0.006 \\
				&       & 10 & 0.009  & 0.008  & 0.008 & 0.006 &  &     &       & 10 & 0.007  & 0.007  & 0.007 & 0.007 \\
				& 0.05  & 1  & 0.047  & 0.045  & 0.045 & 0.044 &  &     & 0.05  & 1  & 0.048  & 0.047  & 0.048 & 0.047 \\
				&       & 5  & 0.050   & 0.048  & 0.050  & 0.050  &  &     &       & 5  & 0.045  & 0.045  & 0.045 & 0.044 \\
				&       & 10 & 0.037  & 0.035  & 0.036 & 0.036 &  &     &       & 10 & 0.034  & 0.033  & 0.034 & 0.033 \\
				& 0.10  & 1  & 0.101  & 0.100    & 0.101 & 0.101 &  &     & 0.10  & 1  & 0.102  & 0.100    & 0.100   & 0.009 \\
				&       & 5  & 0.091  & 0.089  & 0.088 & 0.087 &  &     &       & 5  & 0.093  & 0.093  & 0.093 & 0.093 \\
				&       & 10 & 0.084  & 0.083  & 0.083 & 0.082 &  &     &       & 10 & 0.076  & 0.075  & 0.074 & 0.074 \\
				500 & 0.01  & 1  & 0.009  & 0.009  & 0.009 & 0.009 &  & 500 & 0.01  & 1  & 0.008  & 0.008  & 0.008 & 0.008 \\
				&       & 5  & 0.009  & 0.008  & 0.01  & 0.009 &  &     &       & 5  & 0.006  & 0.006  & 0.006 & 0.006 \\
				&       & 10 & 0.010   & 0.01   & 0.011 & 0.011 &  &     &       & 10 & 0.007  & 0.007  & 0.007 & 0.007 \\
				& 0.05  & 1  & 0.045  & 0.044  & 0.043 & 0.042 &  &     & 0.05  & 1  & 0.051  & 0.050   & 0.05  & 0.049 \\
				&       & 5  & 0.047  & 0.046  & 0.047 & 0.047 &  &     &       & 5  & 0.045  & 0.044  & 0.045 & 0.044 \\
				&       & 10 & 0.046  & 0.046  & 0.047 & 0.047 &  &     &       & 10 & 0.044  & 0.044  & 0.044 & 0.044 \\
				& 0.10  & 1  & 0.111  & 0.111  & 0.111 & 0.111 &  &     & 0.10  & 1  & 0.104  & 0.104  & 0.104 & 0.104 \\
				&       & 5  & 0.100    & 0.098  & 0.101 & 0.1   &  &     &       & 5  & 0.094  & 0.093  & 0.094 & 0.093 \\
				&       & 10 & 0.093  & 0.093  & 0.093 & 0.093 &  &     &       & 10 & 0.091  & 0.089  & 0.09  & 0.089 \\ \bottomrule
			\end{tabular}
		\end{adjustbox}
\end{table}
\section{Details of the estimation of the FSAR model}\label{app-FSAR}
The FSAR model was  introduced in \cite{horvath:2020:forward-curve}. The estimation method that we use is based on a least-squares principle and functional principal component analysis (FPCA). 
The FSAR model is defined as
\begin{equation}
	X_{n}(t)=\Phi_{\ell_1}\left(X_{n-\ell_{1}}\right)(t)+\cdots+\Phi_{\ell_k}\left(X_{n-\ell_{k}}\right)(t)+\varepsilon_{n}(t), 
\end{equation}
where each $\ell_i\in\{1,\dots,n\}$ and $\ell_i \ne \ell_j$ for $i\ne j$. Assuming the kernels $\phi_{\ell_i}$ and the observations  $X_i$'s can be well approximated by the first $p$ functional principal components, we then let
\[
\bX_i=\left(\left\langle X_i, \hat{v}_{1}\right\rangle, \ldots,\left\langle X_{i}, \hat{v}_p\right\rangle\right)^{\top} \in \mathbb{R}^{p}.
\]
and 
\begin{align*}
	\mathbf{X}_{L, p,i} &=  \begin{pmatrix}
		\bX_{i-\ell_1} \\
		\bX_{i-\ell_2} \\
		\vdots \\
		\bX_{i-\ell_k} 
	\end{pmatrix} =
	\begin{pmatrix}
		\inp{X_{i-\ell_1}}{\hat{v}_1} \\
		\inp{X_{i-\ell_1}}{\hat{v}_2} \\
		\vdots\\
		\inp{X_{i-\ell_k}}{\hat{v}_p}
	\end{pmatrix} \in \mathbb{R}^{kp}
\end{align*}
Then 
\begin{align*}
	&\bX_{L,p} = \begin{pmatrix}
		\mathbf{X}_{L, p,\ell_k+1}^\top &
		\mathbf{X}_{L, p,\ell_k+2}^\top &
		\cdots  &
		\mathbf{X}_{L, p, N}^\top  
	\end{pmatrix} \\
	&=\begin{pmatrix}
		\inp{X_{\ell_k+1-\ell_1}}{\hat{v}_1} &  \inp{X_{\ell_k+1-\ell_1}}{\hat{v}_2} & \cdots &
		\inp{X_1}{\hat{v}_1} &\cdots &\inp{X_1}{\hat{v}_p} \\
		\inp{X_{\ell_k+2-\ell_1}}{\hat{v}_1} &  \inp{X_{\ell_k+2-\ell_1}}{\hat{v}_2} & \cdots &
		\inp{X_2}{\hat{v}_1} &\cdots &\inp{X_2}{\hat{v}_p} \\
		\vdots \\
		\inp{X_{n-\ell_1}}{\hat{v}_1} &  \inp{X_{n-\ell_1}}{\hat{v}_2} & \cdots &
		\inp{X_{n-\ell_k}}{\hat{v}_1} &\cdots &\inp{X_{n-\ell_k}}{\hat{v}_p}
	\end{pmatrix}\\
	&\in \mathbb{R}^{(n - \ell_k)\times kp},\\
	&\bX_{R,p} = 
	\begin{pmatrix}
		\bX_{\ell_{k}+1}^\top \\
		\vdots \\
		\bX_n^\top
	\end{pmatrix} \\
	&= \begin{pmatrix}
		\inp{X_{\ell_k+1}}{\hat{v}_1} & \inp{X_{\ell_k+1}}{\hat{v}_2} &\dots  & & \inp{X_{\ell_k+1}}{\hat{v}_p} \\
		\vdots \\
		\inp{X_n}{\hat{v}_1} & \inp{X_n}{\hat{v}_2} &\dots  & & \inp{X_n}{\hat{v}_p} \\
	\end{pmatrix} \in \mathbb{R}^{(n-\ell_k)\times p}.
\end{align*}
The least-squares estimator is:
\[
\widehat{\boldsymbol{\Phi}}=\left(\mathbf{X}_{L, p}^{\top} \mathbf{X}_{L, p}\right)^{-1} \mathbf{X}_{L, p}^{\top} \mathbf{X}_{R, p}
=:\left(\begin{array}{c}
	\hat{\boldsymbol{\Phi}}_{\ell_1} \\
	\vdots \\
	\hat{\boldsymbol{\Phi}}_{\ell_k}
\end{array}\right) \in \mathbb{R}^{k p \times p}
\]
Then the estimates of the kernel function $\psi_{\ell_i}(t,s)$ is 
\[
\hat{\phi}_{\ell_i}(t, s) \approx \sum_{j=1}^p\sum_{r=1}^p \hat{\boldsymbol{\Phi}}_{\ell_i}[j, r] \hat{v}_{j}(t) \hat{v}_{r}(s), \quad \hat{\boldsymbol{\Phi}}_{i} \in \mathbb{R}^{p \times p},
\]
where the square bracket $[\cdot, \cdot]$ shows the indices of the matrix.\\
It may be re-written in the matrix form as
\[
\hat{\phi}_{\ell_i}(t, s) \approx \begin{pmatrix}
	\hat{v}_1(t), & \hat{v}_2(t), & \cdots & \hat{v}_p(t)
\end{pmatrix}    
\hat{\boldsymbol{\Phi}}_{\ell_i}\begin{pmatrix}
	\hat{v}_1(s) \\ \hat{v}_2(s) \\ \vdots \\ \hat{v}_p(s)
\end{pmatrix}    . 
\]
The forecasts are then made by using the estimated kernel as 
\[
\hat{X}_i(t) = \int \hat{\phi}_{\ell_1}(t, s) X_{i-\ell_1}(s)ds +\cdots +\int \hat{\phi}_{\ell_k}(t, s) X_{i-\ell_k}(s)ds.
\]
\end{appendices}
\spacingset{1.00} 
\bibliographystyle{chicago}
\bibliography{fSACF}
\end{document}